\documentclass[aps,pra,twocolumn,footinbib,
showpacs,showkeys,longbibliography,nobalancelastpage]{revtex4-2}

\usepackage{graphicx}
\usepackage{indentfirst}
\usepackage{physics}
\usepackage{braket}
\usepackage{float}
\usepackage{amsmath}
\usepackage{amssymb}
\usepackage{verbatim}
\usepackage{epstopdf}
\usepackage{CJK}
\usepackage{esint}
\usepackage{color}
\usepackage{xcolor}
\usepackage{epsfig}
\usepackage{subfigure}
\usepackage{amsfonts}
\usepackage{footmisc}
\usepackage{soul}
\usepackage{scrextend}
\usepackage{multirow}
\usepackage{mathtools}
\usepackage{xr-hyper}
\usepackage[hyperfootnotes=false]{hyperref}

\usepackage[english]{babel}
\usepackage{url}
\usepackage{bm}
\definecolor{darkblue}{rgb}{0,0,0.5}
\hypersetup{
colorlinks=true,
linkcolor=black,
filecolor=black,
citecolor=darkblue,
urlcolor=darkblue,
}

\newcommand{\defeq}{\vcentcolon=}

\newcommand\mc[1]{\mathcal{#1}}
\newcommand\kCPF[1]{\mathcal{U}_{\text{CPF}}^{#1}}
\newcommand\bs[1]{\boldsymbol{#1}}

\newtheorem{theorem}{Theorem}
\newtheorem{lemma}{Lemma}

\newtheorem{defin}{Definition}
\newenvironment{proof}[1][Proof]{\noindent\textbf{#1.} }{\ \rule{0.5em}{0.5em}}

\begin{document}

\title{Analytical Bounds for Dynamic Multi-Channel Discrimination}
\author{Cillian Harney and Stefano Pirandola} 
\affiliation{Department of Computer Science, University of York, York YO10 5GH, United Kingdom}
\begin{abstract}
The ability to precisely discriminate multiple quantum channels is fundamental to achieving quantum enhancements in data-readout, target detection, pattern recognition, and more. Optimal discrimination protocols often rely on entanglement shared between an incident probe and a protected idler-mode. While these protocols can be highly advantageous over classical ones, the storage of idler-modes is extremely challenging in practice. In this work, we investigate idler-free block protocols based on the use of multipartite entangled probe states. In particular, we focus on a class of idler-free protocol which uses non-disjoint distributions of multipartite probe states irradiated over multi-channels, known as dynamic discrimination protocols.  We derive new, analytical bounds for the average error probability of such protocols in a bosonic Gaussian channel setting, revealing idler-free strategies that display performance close to idler-assistance for powerful, near-term quantum sensing applications.
\end{abstract}

\maketitle
A critical setting of quantum hypothesis testing \cite{HelstromQHT, BaeQSD, QSDBergou, QSDChefles, DiscrimQOps, MarcheseQHT} is Quantum Channel Discrimination (QCD), in which a discriminator is tasked with discerning between a collection of quantum channels \cite{QSensing,AdvPhS,FLQCD}. Quantum channels are ubiquitous in their description of physical phenomena. Hence the design of QCD protocols that leverage superiority over optimal classical strategies, offers exciting technological advancement in many fields of quantum sensing; data read-out \cite{PirBD, NairNDS,Tej2013, Spedalieri2012, Zhuang2017_3,Hirota2017, expqread,QEnhDataCl}, target-detection \cite{Lloyd2008, Tan2008, Shapiro2009, Zhang2013, Zhang2015, Barzanjeh2015, Dallarno2012, Zhuang2017_1, Zhuang2017_2, LasHeras2017, DePalma2018, NairIllum,Lopaeva2013, Barzanjeh2020}, cryptography \cite{GSCryptoQR, AmpDampB} and more.\par

Whilst binary QCD has been an avenue of  intense research ever since Helstrom's seminal work in the late 1960's, much less is known about multi-channel discrimination protocols (MQCD). In this setting, one must distinguish between multiple quantum channels, or \textit{channel patterns}, which consist of an arrangement of $m \geq 2$ channels. MQCD is then the task of classifying between different channel patterns, from a potentially large and non-uniform pattern space. There is also an enormous space of generally adaptive quantum protocols that can be employed for such a task, making the determination of optimal protocols extremely difficult. \par

A new insight has been obtained with the advent of Channel Position Finding (CPF) \cite{EntEnhanced}. CPF describes the multi-ary discrimination task of locating a single target channel hidden amongst background channels. It has been recently shown that special classes of multi-channels satisfying joint teleportation covariance can be optimally discriminated non-adaptively \cite{UltPrec2017,ULMCD}, using quantum probes consistent of a signal-mode and idler-mode which are maximally entangled. The signal-mode interacts with a channel, while the idler-mode is perfectly preserved from decoherence. With a particular focus on the bosonic quantum channel setting, these probes take the form of Two Mode Squeezed Vacuum (TMSV) states, and can be used to define optimal protocols for the discrimination of Gaussian Phase Insensitive (GPI) channels. Furthermore, the CPF framework provides a stepping stone to increasingly complex discrimination problems, unveiling quantum enhancements in future pattern recognition applications \cite{PatternRecog,ThermalPatt}.\par

Unfortunately, the necessity of ancillary idler-modes poses a considerable practical challenge, as it is never possible to offer perfect protection from decoherence. The alleviation of loss or noise may be achieved via quantum memories; however the design of quantum memories with simultaneously long storage time and memory efficiency is a prominent theoretical/experimental hurdle, and active area of research \cite{Lvovsky_QMem_2009, Jensen_QMem_2010,Cho_QMem_2016,Wang_QMem_2019}. It is thus desirable to ask: Do there exist high performance idler-free discrimination protocols that offer quantum advantage?\par

This question has been recently addressed through the use of multipartite entangled probe states \cite{Jason_IdlerFree, MP_IdlerFree}. While mostly unexplored, the alternative utilisation of multipartite probe-mode entanglement has been shown to offer expedient quantum advantage, opening new doors to practically feasible multi-channel discrimination. Research in \cite{MP_IdlerFree} showed that idler-free, multipartite entangled probe states can be designed that are superior over classical protocols for large regions of channel parameter space, in some cases matching the efficacy of idler-assistance. Yet, the enormous space of possible multipartite probing structures and complexity of performance benchmarking meant that much of this analysis was numerical, focussing on smaller pattern dimensions.\par

Here, we improve upon this progress by deriving analytical discrimination error bounds for unassisted discrimination protocols using multipartite entangled probe states. 
We investigate unassisted \textit{dynamic discrimination protocols}, where multi-channels are probed by multipartite input states which reconfigure and vary their spatial probing domains over the course of multiple rounds.  
Remarkably, the variation of probing domains provide an intrinsic error-correcting behaviour for non-adaptive, idler-free protocols. 
By carefully engineering these protocols with entangled CV-GHZ states (the bosonic analogue to the discrete-variable GHZ state), then we derive analytical error probability bounds for the tasks of bosonic quantum reading and environment localisation \cite{PirBD,OptEnvLoc}. These are problems which are relevant to accomplishing future quantum enhancements in optical data-readout/transfer, quantum cryptography, and pattern recognition.\par

This paper proceeds as follows: In Section~\ref{sec:Prelims} the tasks of quantum pattern recognition and multi-channel discrimination are introduced, and channel patterns are specified to bosonic Gaussian quantum channels. In Section~\ref{sec:Unassisted} we review recent developments in the regime of unassisted protocols, and specifically define fixed and dynamic block protocols which exploit multipartite quantum states. 
A general method for the derivation of analytical error bounds associated with unassisted fixed and dynamic block protocols is developed in Section~\ref{sec:Methods}. Section~\ref{sec:DDP} then applies the machinery of the previous sections to produce the key results of this paper; providing new, compact error bounds for high-performance dynamic block protocols. Finally, Section~\ref{sec:Concl} concludes our work, with discussions of future investigative paths.

\section{Preliminaries\label{sec:Prelims}}
\subsection{Quantum Pattern Recognition}

Let us define a binary quantum channel pattern as an $m$-length multi-channel, where each channel is characterised as a background channel ($B$) or a target quantum channel ($T$),
\begin{align}
\mc{E}_{\bs{i}} &\defeq \mc{E}_{i_1} \otimes \mc{E}_{i_2} \otimes \ldots \otimes \mc{E}_{i_m},\\ &=  \bigotimes_{k=1}^m \mc{E}_{i_k},\> i_k \in \{B,T\}.
\end{align}

This object is best interpreted as a quantum mechanical representation of an $m$-pixel \textit{image} described by a sequence of binary variables $\bs{i} = i_1, \ldots, i_m$. Each pixel $i_k \in \{B,T\}$ can be described by a background channel $\mc{E}_B$ or target channel $\mc{E}_{T}$ which describes a physical property of the environment that is being investigated; such as environmental temperature, or reflectivity of a surface. As such, quantum patterns can be used to provide a language for pattern recognition in terms of generally quantum resources.
{Formally, a channel pattern can also be represented as \textit{multiset} of binary variables which captures the channel pattern $\bs{i} =\{ i_1, \ldots, i_m\}$. Treating channel patterns as multisets is very useful for concepts and operations used later in this work.}
\par

A channel pattern describes a single instance of a binary quantum multi-channel. For multi-channel discrimination, we are concerned with discriminating ensembles of multi-channels from one another. For this reason, we may define an \textit{image space} $\mc{U} = \{\bs{i}_1, \bs{i}_2, \ldots, \bs{i}_N\}$ as a collection of channel patterns. Then, given an arbitrary image space $\mc{U}$ it possible to define binary quantum pattern recognition as the task of discriminating a statistical ensemble of quantum multi-channels $\{ p_{\bs{i}} ; \mc{E}_{\bs{i}} \}_{\bs{i}\in \mc{U}}$, such that $\{\mc{E}_{\bs{i}}\}_{\bs{i}\in\mc{U}}$ is a set of possible multi-channels, each of which occur according to the probability distribution $\{p_{\bs{i}} \}_{\bs{i}\in\mc{U}}$. 

There are a number of important image spaces that can be used to model physical scenarios for future quantum technologies. Of course, it is essential to consider the uniform binary image space $\bs{i} \in \mc{U}_{m}$ which collects all $2^m$ possible combinations of binary $m$-channels. This image space has been studied from the perspective of optical and thermal pattern recognition \cite{PatternRecog,ThermalPatt} (and can also be referred to as quantum barcode decoding). This is the most difficult quantum pattern recognition scenario, since a discrimination protocol must be able to distinguish between all possible binary patterns. 

\begin{figure}[t!]
\includegraphics[width=0.9\linewidth]{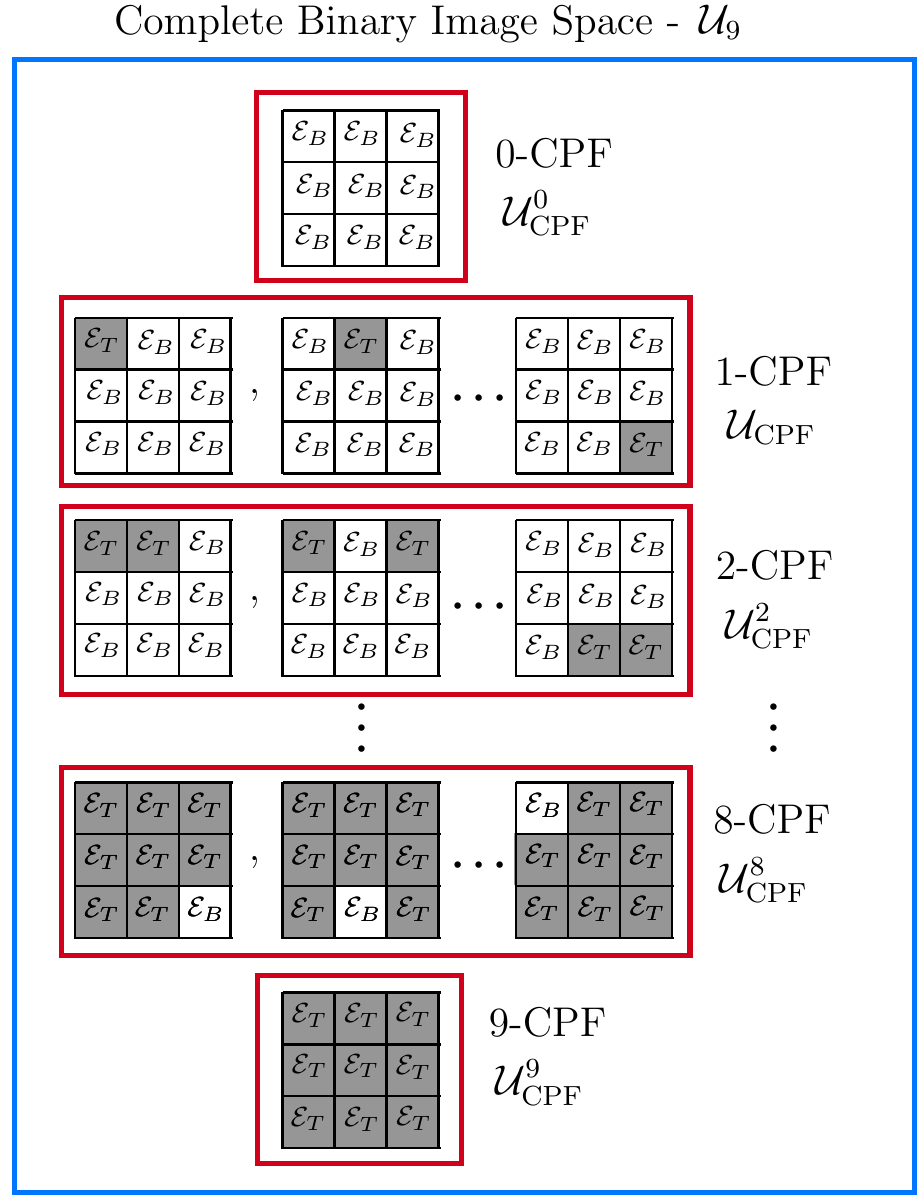}
\caption{Examples of CPF image spaces for $m=9$-length binary channel patterns. Furthermore, the complete set of all binary channel patterns can be decomposed into the sets of $u$-CPF image spaces.}
\label{fig:CPFexplain}
\end{figure}

The challenge of channel position finding (CPF) is capture via a smaller, more specific image space \cite{EntEnhanced}. Indeed, $m$-channel CPF is the task of identifying a single target channel which is hidden amongst $(m-1)$ background channels. This can be captured by considering an image space $\mc{U}_{\text{CPF}}$ which contains all the possible $m$-length patterns that possess exactly one target channel, e.g.~for $m=3$, 
\begin{align}
\begin{aligned}
\mc{U}_{\text{CPF}} &= \{\{B,B,T\},\{B,T,B\},\{T,B,B\}\},\\
&= \{\{0,0,1\},\{0,1,0\},\{1,0,0\}\},
\end{aligned}
\end{align}
where we can equivalently represent channel patterns as binary strings, such that $i_j = 0$ represents a background channel, and $i_j= 1$ represents a target channel. More generally, we can investigate a $u$-CPF image space, which is the set of all $m$-length channel patterns such that each pattern contains precisely $u$ target channels. This image space $\mc{U}_{\text{CPF}}^{u}$ contains a total of $C_m^u = \frac{m!}{u!(m-u)!}$ patterns (where $C_m^u$ is the binomial coefficient). Clearly, when $u=1$, we retrieve the original task of single target CPF. \par

CPF is an extremely useful platform for investigating quantum pattern recognition. It immediately models a number of relevant tasks within target detection, quantum enhanced data-readout, and much more. Furthermore, the construction of effective protocols for multi-channel discrimination in this framework can act as a primitive for studying more complex tasks, and inventing high performance protocols in more challenging settings. In Fig.~\ref{fig:CPFexplain} we illustrate the CPF image spaces and how they decompose the complete binary pattern space. 

\subsection{Multi-Channel Block Protocols}

The most general adaptive discrimination protocol is described as a quantum comb \cite{ChiribellaComb}. This is a quantum circuit board that possesses $M$ ``slots"  which may be filled with instances of $m$-length channel patterns $\mc{E}_{\bs{i}}$. The comb itself is a register with an arbitrary number of quantum systems, initially prepared in some state $\rho_{0}$. Each slot in the comb offers an opportunity to interrogate a multi-channel using quantum systems from the register. Before and after each pattern instance, the discriminator may perform arbitrary joint quantum operations (QOs), which can be assumed to be trace-preserving. This allows for the exploitation of unlimited entanglement shared between input and output, and feedback that can be used to adaptively optimise subsequent pattern interactions. After $M$ adaptive probings, the comb is in its final state $\rho_{\bs{i}}^{M}$, and is subject to an optimal, joint POVM $\{\Pi_{\bs{i}^{\prime} }\}_{\bs{i}^{\prime} \in \mc{U}}$. This result of this measurement can be classically post-processed to infer the channel pattern with some probability of error. Let us denote general adaptive protocols by the label $\mc{P}$, and assuming that any pattern $\mc{E}_{\bs{i}}$ occurs with probability $p_{\bs{i}}$ then the average classification error probability is given by
\begin{equation}
p_{\text{err}}(\mc{P}) \defeq  \sum_{\bs{i} \neq \bs{i}^{\prime} \in \mc{U}} p_{\bs{i}} \text{Tr}\left[ {\Pi}_{\bs{i}^{\prime}} \rho_{\bs{i}}^{M}\right] . \label{eq:p_err_gen}
\end{equation}
Here, the sum runs over all unequal channel patterns in the image space. The generality of $\mc{P}$ makes it very difficult to determine optimal strategies and their performance. Despite advancements in this pursuit \cite{ULMCD,UltLimProDisc}, in practice it is much easier to employ non-adaptive protocols.\par

Block protocols $\mc{B}$ present an extremely important sub-class of discrimination protocol, in which channel patterns are irradiated with $M$ identical and independent copies of some input probe state $\rho^{\otimes M}$.
The classification error probability is now,
\begin{equation}
p_{\text{err}}(\mc{B}) \defeq  \sum_{\bs{i} \neq \bs{i}^{\prime} \in \mc{U}} p_{\bs{i}} \text{Tr}\left[ {\Pi}_{\bs{i}^{\prime}} \rho_{\bs{i}}^{\otimes M}\right] . \label{eq:p_err_block}
\end{equation}
If this protocol makes use of entangled, ancillary idler-modes then it is a {block-assisted protocol} $\mc{B}^{\text{a}}$. Otherwise, it is a {block-unassisted protocol} $\mc{B}^{\text{u}}$. Block-assisted protocols are often highly effective, but demand perfect protection of idler-modes in order to maintain entanglement with the signal-mode. Such protocols also rely on quantum memories in order to store and preserve their idlers, which poses a significant challenge for near-term quantum technologies. \par

We can benchmark multi-channel discrimination performance via general, fidelity-based bounds from Refs.~\cite{LBMont, UBPGM}. Let us define the Bures fidelity as,
\begin{equation}
F(\rho,\sigma) = \| \sqrt{\rho}\sqrt{\sigma} \|_1 = \text{Tr}\sqrt{\sqrt{\rho}\sigma\sqrt{\rho}}.
\end{equation}
Now consider an $m$-length channel image space $\bs{i}\in \mc{U}$ that generates the ensemble of multi-channels $\{p_{\bs{i}}, \mc{E}_{\bs{i}}\}_{\bs{i}\in\mc{U}}$ with prior distribution $\{ p_{\bs{i}} \}_{\bs{i}\in\mc{U}}$. Assuming the use of generic, $M$-copy input states $\rho^{\otimes M}$, then the error probability $p_{\text{err}}$ of this discriminatory protocol is bounded by,
\begin{align}
p_{\text{err}} &\geq \frac{1}{2} \sum_{\bs{i}\neq\bs{i}^{\prime}} p_{\bs{i}} p_{\bs{i}^{\prime}} F^{2M}(\rho_{\bs{i}},\rho_{\bs{i}^{\prime}}) ,\label{eq:LB}\\
p_{\text{err}} &\leq \sum_{\bs{i}\neq\bs{i}^{\prime}} \sqrt{p_{\bs{i}} p_{\bs{i}^{\prime}}} F^{M}(\rho_{\bs{i}},\rho_{\bs{i}^{\prime}}), \label{eq:UB}
\end{align}
where we have used the multiplicativity of the fidelity, $F(\rho^{\otimes M},\sigma^{\otimes M})  = F^M(\rho,\sigma)$ to simplify our bounds. 

Throughout this work we focus on deriving bounds for the worst case discrimination scenario, such that no \textit{a priori} information is known about the probability distribution $\{ p_{\bs{i}}\}_{\bs{i}\in \mc{U}}$. That is, we assume that all channel patterns may occur with equal probability, $p_{\bs{i}} = 1/|\mc{U}|, \forall \bs{i}\in\mc{U}$. In this case, it is useful to define a \textit{total error quantity} for an arbitrary block protocol,
\begin{equation}
D[\mc{U},M] \defeq \sum_{\bs{i}\neq\bs{i}^{\prime}\in\mc{U}} F^{M}(\rho_{\bs{i}},\rho_{\bs{i}^{\prime}}), \label{eq:GenB}
\end{equation}
such that the error probability bounds in Eqs.~(\ref{eq:LB}) and (\ref{eq:UB}) can be more succinctly written as
\begin{equation}
\frac{1}{2|\mc{U}|^2} D[\mc{U},2M] \leq p_{\text{err}} \leq \frac{1}{|\mc{U}|} D[\mc{U},M] .
\end{equation}

\subsection{Bosonic Gaussian Quantum Patterns \label{sec:BQPats}}
In this work, we focus on the discrimination of Gaussian phase insensitive (GPI) channels. These are a family of bosonic quantum channels ubiquitous in quantum metrology, communications, and computation. They are parameterised by transmissivity $\tau$ and induced environmental noise $\nu$ \cite{GaussRev, SerafiniCV}. GPI channels preserve the Gaussianity of input states, and are described by symplectic transformations on the covariance matrix (CM) and first moments vector of a Gaussian state. If $V$ is an  CM of a single-mode Gaussian state (with zero first moments), then its transformation under the action of a GPI channel with transmissivity $\tau$ and environmental noise $\nu$ is given by
\begin{equation}
V \mapsto (\sqrt{\tau} I) \>V\> (\sqrt{\tau} I)^T +  \nu I, \label{eq:smCMtrans}
\end{equation}
where $I$ is the $2\times 2$ identity matrix, and $T$ denotes the transposition operator. Sequences of GPI channels with different transmissivity/noise properties can therefore be used to construct quantum channel patterns.\par
Setting $\tau = \eta$, and $\nu = \frac{1- \eta}{2}$ such that $\eta \in (0,1)$, describes a bosonic pure-loss channel. The discrimination of bosonic loss is integral to the task of \textit{quantum reading} \cite{PirBD,PatternRecog}, in which a user is tasked with retrieving classical data from an optical memory, using quantum resources and measurements to enhance performance. Classical information can be stored in cells of different transmissivity, $\eta_j \in \{\eta_B,\eta_T\}$. Constructing practical, quantum-enhanced protocols for multi-lossy channel discrimination is invaluable for high-rate data-readout/transfer, and optical pattern recognition.\par
Alternatively one may investigate patterns constituent of thermal-loss/amplifier or additive-noise channels, in which environmental noise is considered. In this work we focus on the idealised absence of loss, describing a Gaussian additive-noise channel such that $\tau = 1$ and $\nu > 0$. The classification of channels with identical transmissivities ($\tau_B = \tau_T$) but different noise properties ($\nu_j \in \{\nu_B,\nu_T\}$) is known as \textit{environment localisation} \cite{OptEnvLoc}. This has relevant applications in thermal imaging, and eavesdropper identification in multi-mode communication channels.\par
In order to describe the transformation induced by quantum channel patterns on global, $m$-mode Gaussian states, let us define the useful matrix function 
\begin{equation}
I_{[a]_{\bs{i}}} \defeq \bigoplus_{j=1}^m \begin{pmatrix} a_{i_j} & 0 \\ 0 & a_{i_j} \end{pmatrix},
\end{equation} 
where $a$ is some physical channel property that varies with respect to position in the pattern, $\bs{i}$. If $V$ is an $m$-mode CM of a Gaussian state (with zero first moments), then transformation under the action of $\mc{E}_{\bs{i}}$ is given by
\begin{equation}
V_{\bs{i}} = (I_{[\sqrt{\tau}]_{\bs{i}}}) \>V\> (I_{[\sqrt{\tau}]_{\bs{i}}})^T +  I_{[\nu]_{\bs{i}}}. \label{eq:CMtrans}
\end{equation}\par

\section{Unassisted Multi-Channel Discrimination\label{sec:Unassisted}}
\subsection{Unassisted Block Protocols}

Since the use of idlers is not always practical, it is interesting to ask how we may design a quantum enhanced block protocol without them. That is, can we identify a block-unassisted protocol $\mc{B}^{\text{u}}$ that can also provide quantum advantage over the best classical protocol? To do this, one may introduce entanglement between \textit{multipartite signal states} rather than relying on idler-assisted entanglement. An important question arises when employing multipartite states for multi-channel discrimination: How should entangled probe states be distributed over the channel pattern in order to enhance the block protocol? For an $m$-length channel pattern there are many ways in which these probe states can be distributed. We refer to \textit{probe-domains} as regions of a channel pattern over which multipartite probe states are irradiated. 

Consider an $m$-length channel pattern. A block-unassisted protocol using multipartite states allocates an $M$-copy, $n$-mode probe state to interact with a subset of $n \leq m$ channels whose labels are contained within the set $\bs{s} = \{s_1,\ldots,s_n\},~s_i \in [1,m]$. In turn, we can describe a probe sub-state $\rho_{\bs{s}}^{\otimes M}$ which is potentially entangled over the probe-domain $\bs{s}$ but is separable with respect to any subset of channels that are not contained in $\bs{s}$. Considering a number of multipartite states irradiated over different probe-domains, we may define an $n$-partite discrete distribution (or partition set) of channel sub-patterns \begin{equation}
 \mc{S} = \bigcup_{j=1}^n \{ \bs{s}_j \},~ \exists \> j \text{ s.t } k \in \bs{s}_j, \forall k \in [1,m],
\end{equation} 
where the condition on the RHS demands that each channel in the pattern is contained in at least one subset. One can then define a global, unassisted, multipartite probe state
\begin{equation}
\rho_{\mc{S}}^{\otimes M} = \bigotimes_{j=1}^n \rho_{\bs{s}_j}^{\otimes M}.
\end{equation}
As introduced in Ref.~\cite{MP_IdlerFree}, we can identify two unique regimes of unassisted block protocols with multipartite probe states: Fixed block protocols and dynamic block protocols. These strategies differ in their disjointedness properties of the probe-domain distribution.

\subsection{Fixed Block Protocols}

If the probe-domain distribution is disjoint (which we label via the subscript $\mc{S}_{\text{d}}$) then no two subsets are permitted to share the same channel. The probe-domain distribution takes the form 
\begin{equation} 
 \mc{S}_{\text{d}} = \bigcup_{j=1}^n \{ \bs{s}_j \}, \text{ such that ${\bs{s}_j \cap \bs{s}_k} = \varnothing$, $\forall j,k$.}
\end{equation}
{where the condition on the RHS asserts the pairwise disjointedness of all subsets of channel labels in $\mc{S}_{\text{d}}$}. Again we demand that all channels labels are accounted for, but now each label is strictly contained in a single subset $\bs{s}_j$.\par
Since all probe-domains are disjoint, then all the sub-states $\rho_{\bs{s}_j}^{\otimes M}$ can be simultaneously irradiated on their respective sub-regions of the channel pattern without any overlaps. We may describe this strategy as a \textit{fixed block protocol}, $\mc{B}_{\text{fix}}^{\text{u}}$, inspired by its operational interpretation. Disjointedness in the distribution $\mc{S}$ implies that each probe sub-state is \textit{fixed} on a specific region of channels over the complete course of discrimination.\par

\subsection{Dynamic Block Protocols}

If the probe-domain partition set is non-disjoint, then there will exist subsets $\bs{s}\in\mc{S}$ that share some channels. That is, multipartite probe states may overlap over sub-regions of the channel pattern, and the overall block-unassisted protocol must be modified to address this property. Let us relabel the non-disjoint probe-domain distribution $\mc{S}_{\text{nd}}$ in order to make clear the distinction between disjoint and non-disjoint probe-domain distributions.

A non-disjoint distribution can always be decomposed into an $r$-length sequence of disjoint sets. More precisely,
\begin{equation}
 \mc{S}_{\text{nd}} = \bigcup_{i=1}^r \mc{S}_{\text{d}}^i  = \bigcup_{i=1}^r  \bigcup_{\bs{s} \in \mc{S}_{\text{d}}^i} \{ \bs{s} \}
\end{equation}
such that $\mc{S}_{\text{d}}^i$ denotes a disjoint sub-collection of probe-domains. Each $\mc{S}_{\text{d}}^i$ need not contain all channel labels $1,\ldots,m$, however the overall distribution $\mc{S}_{\text{nd}}$ must account for all channels. \par

This decomposition informs us that a non-disjoint distribution of multipartite probes can be irradiated on a multi-channel over the course of $r$ rounds of disjoint pattern interaction. That is, one can apply multipartite probe states with overlapping domains at different disjoint rounds in a protocol, contributing to a single complete round of discrimination.
This defines a \textit{dynamic block protocol}, $\mc{B}_{\text{dy}}^{\text{u}}$, i.e.~since the probe configuration ``moves" throughout the protocol, it can be described as dynamic (see Ref.~\cite{MP_IdlerFree} for more details).

The investigation and utilisation of dynamic protocols require some important clarifications. Firstly, one must be careful when comparing $M$-copy probe state resources within fixed/dynamic protocols. A disjoint probe-domain distribution using $M$-copy probe states will interact with each channel precisely $M$ times. However, a non-disjoint distribution can have overlapping probe-domains, therefore some channels may be probed more than $M$ times over the course of $r$ disjoint rounds of pattern interaction. For this reason, we must instead study the \textit{average channel use} $\bar{M}$ which refers to the average number of times a channel is probed across a complete, $M$-copy block protocol. For an arbitrary probe-domain distribution $\mc{S}$, the average channel use is simply
\begin{equation}
\bar{M}(\mc{S}) = \frac{M}{m} \sum_{\bs{s}\in\mc{S}} |\bs{s}|.
\end{equation}
Clearly, if we utilise $M$-copy probes in a fixed protocol we will return $\bar{M}(\mc{S}_{\text{d}}) = M$, since there are no overlaps and thus $\sum_{\bs{s}\in\mc{S}_{\text{d}}} |\bs{s}| = m$. But a dynamic protocol with overlapping channels will have $\bar{M}(\mc{S}_{\text{nd}}) \geq M$. Therefore, when comparing block protocols, we demand that they have an identical average channel use. By fixing the average channel use $\bar{M}$, we can then compute the appropriate number of probe copies required for a fair comparison between different protocols,
\begin{equation}
M(\mc{S}) = \frac{m}{\sum_{\bs{s}\in\mc{S}} |\bs{s}|} \bar{M}.
\end{equation}\par
The average channel use provides a much more accurate interpretation of channel interaction for multipartite block protocols. Interestingly, non-disjoint probe-domain distributions can invoke fractional numbers of probe-copies $ 0 < M < 1$ while still using an average channel use $\bar{M} = 1$. This has the operational interpretation of probabilistically interacting with the quantum channels in each disjoint round. 
 
 \subsection{Dynamic to Fixed Block Protocol Transformation \label{sec:Transf}}
 
 Given an image space $\mc{U}$ and an ensemble $\{ p_{\bs{i}} ; \mc{E}_{\bs{i}} \}_{\bs{i}\in\mc{U}}$ that defines a multi-channel discrimination problem, we can always utilise the upper and lower bounds on the error probability of classification from Eqs.~(\ref{eq:LB}) and (\ref{eq:UB}) to benchmark their performance. In the case of unassisted block protocols, this requires that we choose a probe-domain distribution $\mc{S}$, a class of input state $\rho$ to build multipartite probe states, and a number of probe copies $M = M(\mc{S})$. In many cases, the error bounds can then be readily computed.\par
 
 For dynamic block protocols with a non-disjoint $\mc{S}_{\text{nd}}$, a further modification is required to investigate the error bounds. An overlapping channel within $\mc{S}_{\text{nd}}$ can be probed by two independent and separable quantum states, irradiated over unique probe-domains. These interactions will happen at different disjoint rounds in the dynamic block protocol. We cannot treat these instances as channel interactions within the same Hilbert space, because they are completely separable, taking place at different times.\par
 
 To deal with this, we can invoke a transformation of the dynamic block protocol into a fixed representation. All instances of overlapping channels can be considered to introduce a new quantum channel which is an identical copy of its original, and is concatenated onto the original channel pattern. For a non-disjoint probe-domain distribution $\mc{S}_{\text{nd}}$ with $m_{\text{ov}}$ overlaps, an $m$-length channel pattern will be mapped to a $(m+m_{\text{ov}})$ modified channel pattern with $m_{\text{ov}}$ ``copy channels". More precisely, this mapping follows,
 \begin{equation}
 \bs{i} = \{i_1,\ldots,i_m\} \mapsto \bs{\nu_i} = \biguplus_{\bs{s}\in\mc{S}_{\text{nd}}} \{ i_j \}_{j\in\bs{s}}.
 \end{equation}
 Here, we have denoted $\bs{\nu_i}$ as the modified quantum pattern, {and utilised the multiset union operator $\uplus$ which permits the concatenation of repeated channel labels in the modified image. For example, for $m=4$ length channel patterns $\bs{i} = \{i_1,i_2,i_3,i_4\}$ and a probe-domain distribution $\mc{S}_{\text{nd}} = \{\{1,2,3\}, \{2,3,4\}\}$, modified channel patterns would admit the form $\bs{\nu_i} = \{i_1,i_2,i_3,i_2,i_3,i_4\}$.} By iterating over all patterns in the image space, a modified image space is constructed $\mc{U} \rightarrow \{ \bs{\nu_i} \}_{\bs{i}\in\mc{U}}$.\par
 
In this way a dynamic block protocol can be equivalently represented as a fixed block protocol, with the simultaneous irradiation of all the sub-states on a modified image space. The multi-channel ensemble that embodies a discrimination problem is also transformed $\{ p_{\bs{i}} ; \mc{E}_{\bs{i}} \}_{\bs{i}\in\mc{U}} \mapsto \{ p_{\bs{i}} ; \mc{E}_{\bs{\nu_i}} \}_{\bs{i}\in\mc{U}}$, where the modified channel patterns take the form,
\begin{equation}
\mc{E}_{\bs{i}} \mapsto \mc{E}_{\bs{\nu_i}} = \bigotimes_{\bs{s}\in\mc{S}_{\text{nd}}} \bigotimes_{j\in\bs{s}} \mc{E}_{i_j}.
 \end{equation}
Thanks to this transformation from dynamic to fixed protocol, we can now conveniently analyse the error bounds from Eqs.~(\ref{eq:LB}) and (\ref{eq:UB}). Assuming the use of a global, unassisted probe state in accordance with a generally non-disjoint $\mc{S}$,
\begin{equation}
\rho_{\mc{S}} = \bigotimes_{j=1}^n \rho_{\bs{s}_j} \rightarrow \rho_{\mc{S},\bs{\nu_i}}= \mc{E}_{\bs{\nu_i}}( \rho_{\mc{S}}),
\end{equation}
then its error probability of classification is bounded by
\begin{align}
p_{\text{err}} &\geq \frac{1}{2} \sum_{\bs{i}\neq\bs{i}^{\prime}\in\mc{U}} p_{\bs{i}} p_{\bs{i}^{\prime}} F^{2M}(\rho_{\mc{S},\bs{\nu_i}},\rho_{\mc{S},\bs{\nu}_{\bs{i}^{\prime}}}) ,\label{eq:modLB}\\
p_{\text{err}} &\leq \sum_{\bs{i}\neq\bs{i}^{\prime}\in\mc{U}} \sqrt{p_{\bs{i}} p_{\bs{i}^{\prime}}} F^{M}(\rho_{\mc{S},\bs{\nu_i}},\rho_{\mc{S},\bs{\nu}_{\bs{i}^{\prime}}}). \label{eq:modUB}
\end{align}
 For more insight into unassisted multi-channel discrimination protocols, we refer the reader to Ref.~\cite{MP_IdlerFree}.

\section{Analytical Methods \label{sec:Methods}}

 \subsection{Bosonic Unassisted Protocols}
 
In this work we assume the use of a specific form of multipartite quantum probe for the discrimination of bosonic channel patterns; the CV-GHZ state $\Phi_{\mu}^m$ \cite{multimodeGHZ}. This a fully-symmetric, $m$-mode Gaussian state which (assuming zero first moments) is fully characterised by its covariance matrix (CM),
\begin{equation}
\Phi_{\mu}^m \rightarrow V_{\mu}^m = \begin{pmatrix}
\mu I & \Gamma & \ldots & \Gamma \\
\Gamma  & \mu I & \ldots & \Gamma \\
\vdots &\ddots  & \ddots & \vdots\\
\Gamma & \Gamma & \ldots & \mu I
\end{pmatrix},
\begin{tabular}{ l } 
$\mu \defeq N_S + \frac{1}{2}$,  \\
\\
$\Gamma \defeq \text{diag}(c_1,c_2)$,
\end{tabular}
\end{equation}
where $N_S$ is the mean photon number irradiated on each channel, and $c_1,c_2$ are correlation parameters. Setting 
\begin{equation}
c_1 = -c_2  = c_{\text{max}} \defeq (m-1)^{-1}{\sqrt{\mu^2 - {1}/{4} }}
\end{equation} 
maximises the quantum correlations between each mode, such that all $m$-modes in the global state are entangled. When $m=2$ these become TMSV states, which we denote by $\varphi_{\mu} \defeq \Phi_{\mu}^{m=2}$.
 
When utilised within a multi-channel discrimination setting, CV-GHZ states allow us to exploit inter-probe entanglement to enhance output state distinguishability. These states are extremely versatile since they are simply parameterised by their intensity of squeezing, and can be readily used to construct $n$-partite entangled states. Indeed, given a general (disjoint or non-disjoint) $n$-partite probe-domain distribution $\mc{S}$ we may define a multipartite CV-GHZ state,
\begin{equation}
\Phi_{\mu}^{\mc{S}} = \bigotimes_{j=1}^n \Phi_{\mu}^{\bs{s}_j} \rightarrow V_{\mu}^{\mc{S}} = \bigoplus_{j=1}^n V_{\mu}^{\bs{s}_j}
\end{equation}
where $\Phi_{\mu}^{\bs{s}_j}$ and $V_{\mu}^{\bs{s}_j}$ describe an entangled CV-GHZ sub-state irradiated over the $j^{\text{th}}$ channel sub-pattern $\bs{s}_j$. After interacting with an arbitrary channel pattern $\mc{E}_{\bs{s}}$, the output states from the protocol can be represented via
\begin{equation}
 \Phi_{\bs{i},\mu}^{\mc{S}} \defeq \mc{E}_{\bs{i}}(\Phi_{\mu}^{\mc{S}}) \rightarrow V_{\bs{i},\mu}^{\mc{S}},
\end{equation} 
where $V_{\bs{i},\mu}^{\mc{S}}$ is transformed according to Eq.~(\ref{eq:CMtrans}).\par

\subsection{Fidelity Degeneracies}

In the most general setting, we can always compute the fidelity-based performance bounds from Eqs.~(\ref{eq:LB}) and (\ref{eq:UB}) numerically, by iterating over all the unequal pairs of channel patterns in an image space $\mc{U}$ and computing the fidelity between their respective output states. However, for larger channel patterns and image spaces, this numerical approach can become computationally infeasible. For a deeper understanding of discrimination protocols and greater efficiency, it is useful to unveil analytical properties of these bounds. 

The use of fully-symmetric probe states, like the CV-GHZ states, is extremely useful for simplifying the calculation of these fidelity bounds. When assuming the use of these states, we find that the fidelity is often highly degenerate. That is, a particular output fidelity $F(\Phi_{\bs{i},\mu},\Phi_{\bs{i}^{\prime},\mu})$ is often equivalent to the output fidelity of many other pairs of channel patterns. This is more precisely explained for $m$-mode CV-GHZ states in Theorem~\ref{theorem:Degen1}.

\begin{theorem}\label{theorem:Degen1}
Consider two subsets of $u/v$-CPF image spaces respectively $\mc{X}$ and $\mc{Y}$, where
\begin{gather} 
d_h(\bs{i},\bs{i}^{\prime}) = d > 0,\\
\forall \bs{i} \in \mc{X}\subseteq\mc{U}_{\text{\emph{CPF}}}^u, \bs{i}^{\prime} \in \mc{Y} \subseteq\mc{U}_{\text{\emph{CPF}}}^v.
\end{gather} 
where $d_h(\bs{i},\bs{i}^{\prime})$ is the Hamming distance between patterns $\bs{i},\>\bs{i}^{\prime}$. This means that $\mc{X}$ and $\mc{Y}$ are Hamming distance preserving subsets of $~\mc{U}_{\text{\emph{CPF}}}^u$ and $\mc{U}_{\text{\emph{CPF}}}^v$. Consider two GPI binary channel patterns $\mc{E}_{\bs{i}},\>\mc{E}_{\bs{i}^{\prime}}$ with identical physical properties, and an $m$-mode CV-GHZ probe state $\Phi_{\mu}$. 
It follows that the fidelity
\begin{equation}
F\left[\mc{E}_{\bs{i}}(\Phi_{\mu}),\mc{E}_{\bs{i}^{\prime}}(\Phi_{\mu}) \right] = F(\Phi_{\bs{i},\mu},\Phi_{\bs{i}^{\prime},\mu} ) ,
\end{equation}
is degenerate for all $\mc{X} \ni \bs{i} \neq \bs{i}^{\prime} \in \mc{Y}$. 
\end{theorem}

The proof for this theorem is found in Appendix \ref{sec:FidDegen}. This degeneracy property is enormously useful. If we consider an arbitrary image space $\mc{U}$, then the total number of ways that we can select unequal pairs (and therefore compute potential output fidelities in Eqs.~(\ref{eq:LB}) and (\ref{eq:UB})) of channel patterns is $|\mc{U}|(|\mc{U}|-1)$. For large pattern spaces, this will be enormous. Luckily, Theorem 1 informs us that when using fully symmetric CV-GHZ states as input quantum probes, the number of unique fidelities that may occur when probing an arbitrary image space is dramatically smaller than this. Indeed, this fidelity degeneracy property tells us that if there are exactly $g_{\text{fid}}$ unique output fidelities, typically $g_{\text{fid}} \ll |\mc{U}|(|\mc{U}|-1) $.\par

It is useful to introduce the following convenient notation to compactly describe the unique output fidelities. Let us collect the channel pattern properties that invoke a degenerate output fidelity. We define a pair of \textit{target numbers} $u$ and $v$ as the number of target channels in $\bs{i}$ and $\bs{i}^{\prime}$ respectively. The fidelity degeneracy also depends on the Hamming distance between these patterns, ${d_h}(\bs{i},\bs{i}^{\prime}) = d$. Along with squeezing $\mu$, these parameters fully characterise any potential output fidelity that can be obtained by irradiating $m$-mode CV-GHZ states over $m$-length channel patterns. As such, we can collect them into the following object,
\begin{equation}
[u:v\>|\>d] \mapsto [\text{Target no. tuple$\>|\>$Hamming distance}].
\end{equation}
In doing so, we can succinctly define a unique fidelity function,
\begin{gather}
F_{[u:v|d]} \defeq F(\Phi_{\bs{i},\mu},\Phi_{\bs{i}^{\prime},\mu}),\\
\forall \bs{i}\in\mc{U}_{\text{CPF}}^u, \text{  } \bs{i}^{\prime}\in\mc{U}_{\text{CPF}}^v, \text{ such that } {d_h}(\bs{i},\bs{i}^{\prime}) = d.
\end{gather}
This states the output fidelity between $m$-mode CV-GHZ states which have probed a pair of $m$-length binary channel patterns, which $u$ and $v$ target channels respectively, and which have Hamming distance $d$. 

Let us take an example: Consider a pair of $m=4$ length channel pattern sequences, represented by binary strings,
\begin{gather}
\bs{i} = \{0,0,1,1\},~ \bs{i}^{\prime} = \{1,1,0,1\}.
\end{gather}
%where $i_j = 0$ represents a background channel, and $i_j= 1$ represents a target channel.
There are $u=2$ target channels in $\bs{i}$, and $v=3$ target channels in $\bs{i}^{\prime}$. Meanwhile, their Hamming distance is $d_h(\bs{i},\bs{i}^{\prime}) = 3$. By probing the channel patterns $\mc{E}_{\bs{i}}$ and $\mc{E}_{\bs{i}^{\prime}}$ with identical $m=4$ mode CV-GHZ states $\Phi_{\mu}$, the fidelity between the two possible output states is given by 
\begin{equation} 
F( \Phi_{\bs{i},\mu}, \Phi_{\bs{i}^{\prime},\mu}) = F_{[2:3|4]}.
\end{equation}

More generally, we can collect all unique fidelities into a single mathematical object which we call a \textit{unique fidelity set}. To do so we provide the following definition.

\begin{defin} \emph{(Unique Fidelity Set):} 
A unique fidelity set $\bs{F}_m$ is the set of all non-degenerate and non-unit output fidelities for $m$-mode CV-GHZ states which are generated when probing $m$-length binary channel patterns. That is,
\begin{equation}
\bs{F}_{m} = \left\{ F_{[u:v|d]} ~\Big|~ \substack{v \geq u \in \{0,\ldots,m\}, \\
d \in \{v-u, \ldots, \min\{u+v, 2m - (u+v)\}} \right\}.
\end{equation}
This iterates over all valid combinations of target numbers $u,v$ and Hamming distances $d$ for $m$-mode states and $m$-length binary multi-channels.
\end{defin}
For example, it is easy to see that for $m=2$ (TMSV states irradiated over $m=2$ length channel patterns) the unique fidelity set is simply
\begin{align}
\begin{aligned}
\bs{F}_{2} &= \{ {F}_{[0:1|1]},  \> {F}_{[0:2|2]}, \> {F}_{[1:1|1]}  ,\> {F}_{[1:2|1]}  \},\\
 &=\{ {F}_{0:1},  \> {F}_{0:2}, \> {F}_{1:1}  ,\> {F}_{1:2} \},
\end{aligned}
\label{eq:TMSFid}
\end{align} 
where in the second line we drop the Hamming distance label $d$, since there is only one unique Hamming distance per pair of target numbers.

\subsection{General Analytical Method \label{sec:GenMeth}}

Fidelity degeneracies are at the crux of our analytical method for simplifying discrimination error bounds in the unassisted multi-channel setting. They tell us that when using CV-GHZ states over binary channel patterns, the number of output fidelities required for computing the bounds in Eqs.~(\ref{eq:LB}) and (\ref{eq:UB}) is greatly diminished. 

Consider an $n$-partite probe-domain distribution ${\mc{S} = \{\bs{s}_j\}_{j=1}^{n}}$ which constructs an unassisted probe state to be irradiated over $m$-length channel patterns. Let us at first assume that $\mc{S}$ is disjoint, and thus describes a fixed-protocol. Given an $m$-length channel pattern $\bs{i}$, we can split this up into $n$ \textit{sub-patterns} $\bs{i_s}$, which are used to define regions of multi-channels within the global pattern. That is,
\begin{equation}
 \bs{i} = \bigcup_{\bs{s}\in \mc{S}} \bs{i}_{\bs{s}} = \bigcup_{\bs{s}\in \mc{S}} \{ i_j \}_{j\in \bs{s}}
\end{equation}
Hence, each sub-pattern $\bs{i}_{\bs{s}}$ corresponds to a multi-channel $\mc{E}_{\bs{i}_{\bs{s}}}$ which is probed by a CV-GHZ state. Using this fixed protocol, consider the fidelity between the potential output states associated with probing two binary channel patterns $\bs{i}$ and $\bs{i}^{\prime}$, given by $F( \Phi_{\bs{i},\mu}^{\mc{S}}, \Phi_{\bs{i}^{\prime},\mu}^{\mc{S}})$. We can exploit the multiplicativity of the fidelity in order to write
\begin{align}
F( \Phi_{\bs{i},\mu}^{\mc{S}}, \Phi_{\bs{i}^{\prime},\mu}^{\mc{S}}) &= F\Big( \bigotimes_{\bs{s}\in\mc{S}} \Phi_{\bs{i}_{\bs{s}},\mu}^{\bs{s}}, \bigotimes_{\bs{s}\in\mc{S}} \Phi_{\bs{i}_{\bs{s}}^{\prime},\mu}^{\bs{s}}\Big),\\
&= \prod_{\bs{s}\in\mc{S}} F(\Phi_{\bs{i}_{\bs{s}},\mu}^{\bs{s}},\Phi_{\bs{i}_{\bs{s}}^{\prime},\mu}^{\bs{s}}). \label{eq:DisjFidel}
\end{align}
That is, the total output fidelity is equal to the product of all the output fidelities from the local sub-patterns $\bs{i}_{\bs{s}}$ being probed by $|\bs{s}|$-mode CV-GHZ states. 

Crucially, knowledge of the CV-GHZ state fidelity degeneracies allows us to simplify this product. Indeed, given a probe domain $\bs{s}$, we know the sub-pattern output fidelity must be an element of the unique fidelity set,
\begin{equation}
F(\Phi_{\bs{i}_{\bs{s}},\mu}^{\bs{s}},\Phi_{\bs{i}_{\bs{s}}^{\prime},\mu}^{\bs{s}}) 
\in \bs{F}_{|\bs{s}|}.
\end{equation}
Therefore, it is desirable to obtain a general expression for the total error quantity in Eq.~(\ref{eq:GenB}) which is exclusively constructed in terms of fidelities from the unique fidelity sets. In doing this, we can hugely simplify the error bounds in Eqs.~(\ref{eq:LB}) and (\ref{eq:UB}) by removing the necessity to compute all the possible output fidelities within a potentially very large image space. In this way, we can construct a general tool for all unassisted multi-channel discrimination protocols using fully symmetric input states.

Let us define a counting function $t_{\bs{i}} \defeq \sum_{j=1}^{|\bs{i}|} i_j$ which counts the number of target channels contained within a channel pattern $\bs{i}$. Then given a probe-domain distribution $\mc{S}$, we can equivalently express Eq.~(\ref{eq:DisjFidel}) via
\begin{align}
F( \Phi_{\bs{i},\mu}^{\mc{S}}, \Phi_{\bs{i}^{\prime},\mu}^{\mc{S}}) 
&= \prod_{\bs{s}\in\mc{S}} \prod_{F_{[u:v|d]} \in \bs{F}_{|\bs{s}|}} F_{[u:v|d]}^{\delta_{[u:v|d]}(\bs{i}_{\bs{s}},\bs{i}_{\bs{s}}^{\prime})},
\end{align}
where we define the Kronecker-delta type function which selects the appropriate sub-fidelity from $\bs{F}_{|\bs{s}|}$ associated with a particular sub-pattern pair,
\begin{equation}
\delta_{[u:v|d]}(\bs{j},\bs{j}^{\prime}) = 
\begin{cases}
1,  &[t_{\bs{j}}:t_{\bs{j}^{\prime}} | d_h(\bs{j}, \bs{j}^{\prime})] = [u:v|d],\\
0, &\text{otherwise}.
\end{cases}
\end{equation}
Here, $\delta_{[u:v|d]}$ is used to correctly remove or apply contributions from the unique fidelity set dependent on the nature of the sub-patterns. As a result, we can express the total error quantity from Eq.~(\ref{eq:GenB}) in a very general way,
\begin{equation}
D_{\bs{F}}[\mc{U},M] = \sum_{\bs{i}\neq\bs{i}^{\prime} \in \mc{U}} \prod_{\bs{s}\in \mc{S}}  \prod_{F_{[u:v|d]} \in \bs{F}_{|\bs{s}|}} \hspace{-3mm}  F_{[u:v|d]}^{M \delta_{[u:v|d]}(\bs{i}_{\bs{s}},\bs{i}_{\bs{s}}^{\prime})}, \label{eq:gen_err_quant}
\end{equation}
This simply the total error quantity rewritten exclusively in terms of unique output fidelities. Using Eqs.~(\ref{eq:LB}) and (\ref{eq:UB}), the error probability can then be bounded according to
\begin{equation}
\frac{1}{2|\mc{U}|^2}D_{\bs{F}}[\mc{U},2M] \leq p_{\text{err}} \leq \frac{1 }{|\mc{U}|} D_{\bs{F}}[\mc{U},M].
\end{equation}\par

The above method applies directly to dynamic protocols as well as fixed protocols. If the probe-domain distribution being considered is non-disjoint, one need only perform the dynamic to fixed protocol transformation from Section~\ref{sec:Transf}, retrieving a modified image space $\{\bs{\nu_i}\}_{\bs{i}\in\mc{U}}$. Then the previous methodology can be immediately applied. That is, for generally non-disjoint $\mc{S}$ we can easily adapt Eq.~(\ref{eq:gen_err_quant}) by using the Kronecker-delta type function $\delta_{[u:v|d]}(\bs{\nu}_{\bs{i}_{\bs{s}}},\bs{\nu}_{\bs{i}_{\bs{s}}^{\prime}})$ which acts on modified channel sub-patterns. Hence this approach is completely general for any probe-domain distribution. 

While this appears abstract and inconsequential, it is in fact a very powerful result. Given any image space $\mc{U}$, a probe domain distribution $\mc{S}$ and its unique fidelity sets, we can readily determine the analytical expression ${D}_{\bs{F}}$ used to derive the error bounds of its respective protocol. While this method may still demand a (potentially large) iteration over an image space for complicated probe distributions, it can always be done using symbolic-programming techniques and need only be carried out \textit{once}. The exact, analytical distribution can then be stored as a simple function, and utilised at will. This supersedes any brute force method used to directly compute these bounds.

\subsection{Constrained Probe-Domains}
We obtain greater simplifications if we restrict all probes to a constant domain size, such that $|\bs{s}| = k, \forall \bs{s} \in \mc{S}$. This means that all of the sub-fidelities will belong to the same unique fidelities set, $\bs{F}_k$. It is then useful to define a function that counts the number of times a specific sub-fidelity occurs over all sub-patterns,
\begin{equation}
{c}_{[u:v|d]}(\bs{i},\bs{i}^{\prime}) \defeq \sum_{\bs{s}\in \mc{S}} \delta_{[u:v|d]}(\bs{i}_{\bs{s}},\bs{i}_{\bs{s}}^{\prime}).
\end{equation}
For example, if ${c}_{[0:1|1]}(\bs{i},\bs{i}^{\prime}) = 3$, then over the entire probe-domain distribution, there exist three instances where the output fidelity $F_{[0:1|1]}$ has occurred between sub-patterns in $\bs{i},\bs{i}^{\prime}$. Then the very general Eq.~(\ref{eq:gen_err_quant}) can be simplified because we can exclusively consider fidelities from $\bs{F}_k$,
\begin{equation}
D_{\bs{F}_k}[\mc{U},M] = \sum_{\bs{i}\neq\bs{i}^{\prime} \in \mc{U}} \prod_{F_{[u:v|d]}\in\bs{F}_k} F_{[u:v|d]}^{Mc_{[u:v|d]}(\bs{i},\bs{i}^{\prime})}. \label{eq:Constk}
\end{equation}

\subsection{Two-Mode Probe-Domains}
While it is interesting to explore the use of $(m>2)$-mode CV-GHZ states, widening its domain of entanglement quickly becomes detrimental. As the quantum correlations of the state become more widespread, their intensity diminishes, causing a rapid decay in their utility for quantum sensing \cite{MP_IdlerFree}. Hence there is strong motivation to maximise these correlations by constraining probes to collections of two-mode states. As discussed in Eq.~(\ref{eq:TMSFid}), all two-mode sub-fidelities are unique with respect to Hamming distance, and therefore we can simplify the total error quantity,
\begin{equation}
D_{\bs{F}_2}[\mc{U},M] = \sum_{\bs{i}\neq\bs{i}^{\prime} \in \mc{U}} \prod_{F_{u:v}\in\bs{F}_2} F_{u:v}^{M c_{u:v}(\bs{i},\bs{i}^{\prime})}. \label{eq:TM_EB}
\end{equation}
That is, for any binary multi-channel discrimination problem which exclusively uses unassisted TMSV states, the fidelity-based error bounds will always be a polynomial function of the four unique sub-fidelities, ${\bs{F}_2 =\{ {F}_{0:1},  \> {F}_{0:2}, \> {F}_{1:1}  ,\> {F}_{1:2} \}}$.

There is good reason to utilise exclusively two-mode inputs and dynamic protocols to optimise unassisted performance. The behaviour of each unique sub-fidelity ${F_{u:v} \in \bs{F}_2}$ in Eq.~(\ref{eq:TMSFid}) depends on different environmental parameters, and the mean photon number of the probe sources. Crucially, some of the sub-fidelities in this unique fidelity set are more distinguishable than others, meaning that the ability to discriminate particular pairs of multi-channels can vary significantly from one to the other. 

For instance, in a quantum reading setting (discrimination of multiple bosonic lossy channels) the sub-fidelity $F_{0:2}$ is typically much larger than all of the other sub-fidelities. This means that when irradiating unassisted TMSV states over pairs of bosonic lossy channels, it is difficult to distinguish between the multi-channels $\mc{E}_B \otimes \mc{E}_B$ (zero target channels) and $\mc{E}_T \otimes \mc{E}_T$ (two target channels). This pair of multi-channels is particularly difficult to discriminate for unassisted protocols, and will lead to a decay in overall discrimination performance. It is therefore desirable to avoid directly probing equivalent pairs of channels together, $\mc{E}_B \otimes \mc{E}_B$ or $\mc{E}_T \otimes \mc{E}_T$ with entangled TMSV states. However, we cannot possibly know where these pairs of channels exist before choosing a probe-domain distribution. 

Yet, through the use of non-disjoint probe-domain distributions the contribution of poorly distinguishable sub-fidelities can be minimised. By probing a channel in conjunction with numerous other channels in the pattern, we increase the likelihood of probing a more distinguishable collection of quantum channels. In this way, even if a poorly distinguishable collection of channels is probed, it is highly likely that each channel will also be probed in conjunction with a more distinguishable one. This can be thought of as gathering more ``opinions" on the true nature of channels in the pattern, improving the overall discrimination error rate.\par

Abstractly, dynamic protocols can be interpreted as a form of intrinsic error-correction, achieved by the shifting of probe-domains. The application of a non-disjoint distribution of probe states is mathematically treated by encoding channel patterns $\{\bs{i} \}_{\bs{i}\in \mc{U}}$ into modified, extended patterns $\{\bs{\nu_i} \}_{\bs{i}\in\mc{U}}$, consistent with this distribution. In doing so, the encoded global patterns become easier to distinguish than their original versions, increasing the probability of probing distinguishable multi-channels.\par

\subsection{Dynamic $k$-Local Discrimination Protocols}

Even when constrained to two-mode probe-domains, there are still clearly an enormous number of non-disjoint distributions that can be considered. In general, every binary channel pattern ${\bs{i}}$ can possess a (not necessarily unique) optimal probe-domain distribution for which the error probability of classification is minimised. For many channel patterns in an image space $\mc{U}$, there may exist a vast range of probe-domain distributions which should be used for different multi-channels. However, applying the optimal probe-domain distribution for every channel pattern in the image space would require knowledge of the pattern before the fact; undermining the task of discrimination. 
Instead, it is much more practical and realistic to devise probe-domain distributions which are not necessarily optimal, but perform well when averaged over an entire image space, and utilise entanglement to obtain quantum advantage. 

To this end, we can pragmatically devise a class of dynamic block protocol which we call $k$-Local Nearest Neighbour ($k$-LNN) protocols. These are unassisted block protocols which use exclusively TMSV states such that over the course of a single round of discrimination, every channel in the pattern will be probed in conjunction with $k$ other channels. To ensure that each channel is probed within exactly $k$ probe domains, we can define unique \textit{neighbourhoods}. More precisely, for the $j^{\text{th}}$ channel in an $m$-length multi-channel we can define a $k$-element set of the indices of neighbouring channels, $\mc{N}_{k}(j)$. Neighbours can be defined spatially, e.g.~nearby channels defined on a lattice can be considered neighbours. However, this is an arbitrary choice and neighbours can be defined in whichever manner is most suitable to the application. 

\begin{figure}
\hspace{-0.5cm} (a) \hspace{4.cm} (b)\\
\includegraphics[width=0.475\linewidth]{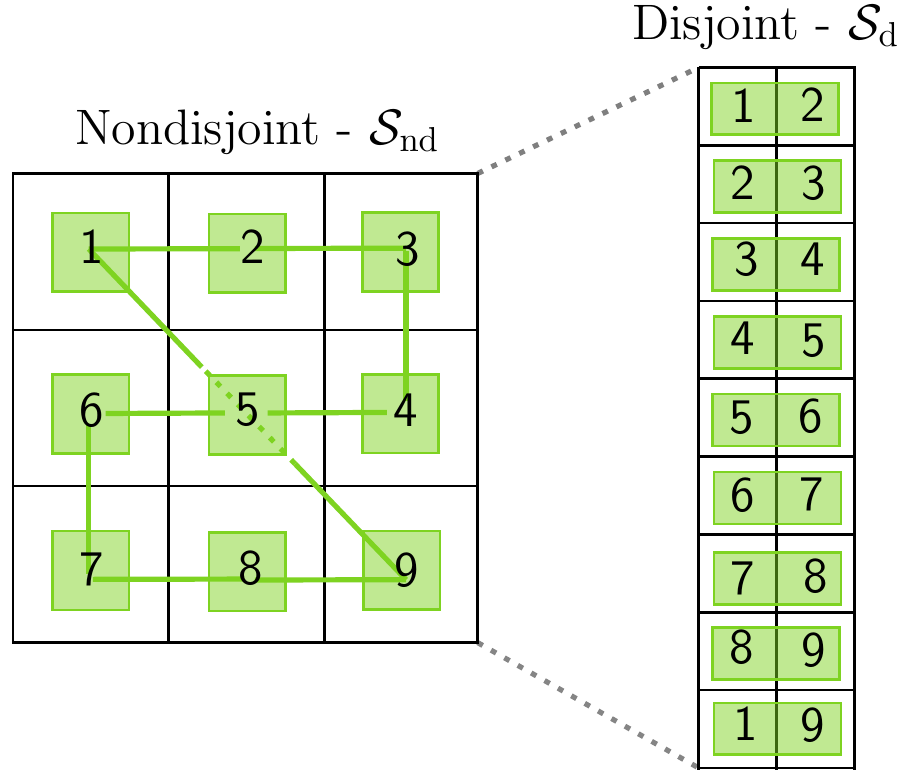}
\includegraphics[width=0.475\linewidth]{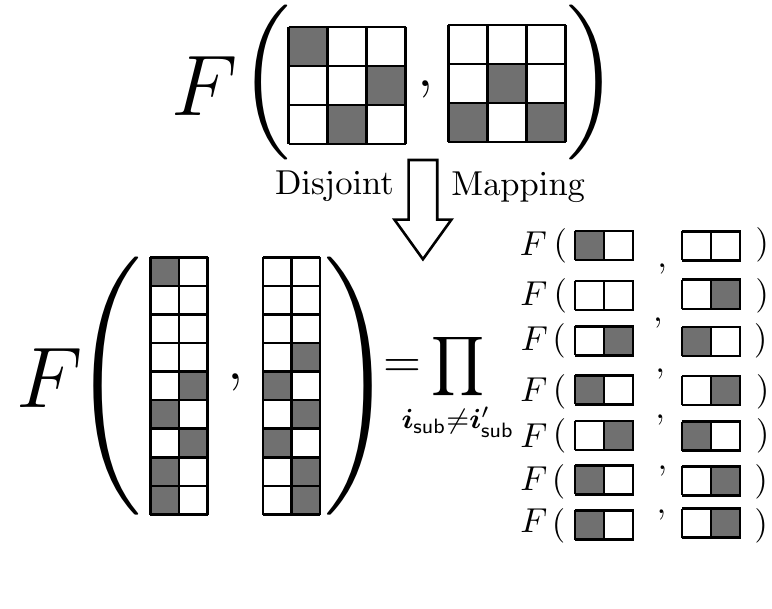}
\caption{(a) The non-disjoint to disjoint mapping for a two-mode constrained, 2-local nearest neighbour (2-LNN) protocol.~(b) Depicts the fidelity computation of two potential patterns undergoing this mapping.}
\label{fig:DisjMap}
\end{figure}

\ifx
Consider the $j^{\text{th}}$ channel in a quantum channel pattern $\bs{i}$, given by $\mc{E}_{i_j}$. For the $j^{\text{th}}$ channel we can define a $k$-element set of the indices of neighbouring channels, $\mc{N}_{k}(j)$. These neighbours can be defined spatially, or as neighbours in a more abstract way. This is used to identify $k$ channels that $\mc{E}_{i_j}$ will be probed in conjunction with, using $k$ pairs of TMSV states. We can then define a non-disjoint probe-domain distribution which contains precisely $km$ elements, such that each channel is probed within $k$ pairs of TMSV states. The non-disjoint partition set of this protocol takes the form,
\begin{equation}
\mc{S}_k^m = \bigcup_{i=1}^m  \bigcup_{j\in \mc{N}_k(i)} \{i,j\} ,
\end{equation}
It is not always possible to probe all channels exact $k$ times, since this depends on the number of channels in the pattern. Thus $k$ must satisfy
\begin{equation}
k = \frac{2n}{m} < m, \text{ for } n \in \mathbb{N}.
\end{equation}
When $k=1$, then all channels are probed using precisely one TMSV state along with one other channel. Therefore the probe-domain distribution is disjoint, and we have a fixed block protocol. For all $2 \leq k \leq m-1$, the protocol is dynamic.
\fi
In the following we provide a formal definition of $k$-LNN protocols.

\begin{defin} \emph{($k$-Local Nearest Neighbour Protocols):} An $M$-copy, $k$-LNN protocol is an unassisted dynamic block protocol which irradiates $m$-length channel patterns using $kmM$ TMSV states.
Each channel is probed in conjunction with exactly $k$-neighbouring channels, where neighbours are defined using a $k$-width neighbourhood function $\mc{N}_k(j)$. The probe-domain distribution follows
\begin{equation}
 \mc{S}_k^m= \bigcup_{i=1}^m  \bigcup_{j\in \mc{N}_k(i)} \big\{ \{i,j\} \big\} ,
\end{equation} 
while valid values of the neighbourhood width $k$ satisfy 
\begin{equation}
{k = \frac{2n}{m} \in \mathbb{N}, \text{ \emph{for} } n\in \mathbb{N}, k < m}.
\end{equation}
As a result, an $M$-copy input state is the product state of many TMSV states, 
\begin{equation}
\bs{\varphi}_{\mu}^{\mc{S}_k^m\otimes M} \defeq \bigotimes_{i=1}^m  \bigotimes_{j\in \mc{N}_k(i)} \varphi_{\mu}^{\{i,j\} \otimes M}.
\end{equation} 
where $\varphi_{\mu}^{\{i,j\}}$ is a TMSV state irradiated over the $i^{\text{th}}$ and $j^{\text{th}}$ channels of the global multi-channel.
\end{defin}

It is not possible to create a $k$-LNN protocol for any $m$-length channel pattern or neighbourhood width, $k$. This is a consequence of an inability to create a probe-domain distribution which contains each channel exactly $k$-times. For example, when $k=1$, then all channels are probed using precisely one TMSV state along with one other channel. Therefore the probe-domain distribution is disjoint, and we have a fixed block protocol. Yet, $k=1$ is only valid when there are an even number of channels. Otherwise, for all $2 \leq k \leq m-1$, the protocol is dynamic.

The average channel use of a $k$-LNN protocol will be $\bar{M} = kM$, since each channel is probed $k$-times per discriminatory round. Furthermore, dynamic protocols described by this non-disjoint partition set $\mc{S}_k^m$ can be mapped to a modified, disjoint format as discussed in Section~\ref{sec:Transf}. For example, for $m=9$ and $k=2$, we generate the probe-domain distribution
\begin{equation}
 \mc{S}_{2}^9 = \big\{ \{ 1,2 \} , \{ 2,3 \} , \ldots , \{ 8,9 \} , \{ 9,1 \} \big\} .
\end{equation}
These modified patterns can now be discriminated disjointly using $km = 18$ TMSV probes. Fig.~\ref{fig:DisjMap}(a) illustrates the dynamic to fixed block protocol mapping for $m=9$ channels, and $k=2$ neighbours. Furthermore, Fig.~\ref{fig:DisjMap}(b) shows how the dynamic to block protocol transformation can then be used to simplify output fidelity calculations into a product of two-mode sub-fidelities.

\section{Results \label{sec:DDP}}
With the previous context and technical ingredients in hand, we are now able to unveil new, analytical error bounds for the bosonic unassisted multi-channel discrimination protocols. In particular, we derive analytical error bounds for dynamic $k$-LNN discrimination protocols which can achieve quantum advantage in the context of quantum reading and environment localisation. We find that not only are these unassisted dynamic protocols better than classical protocols, but in some cases are able to match the performance of idler-assisted protocols for large parameter ranges.

\subsection{Fixed Discrimination Protocols}

Since the probe-domain distribution of a fixed block protocol constrained to two-mode probe states must be disjoint, it can only ever be used over an even number of quantum channels. Hence, in the following we focus on $m$-length channel patterns where $m=2l, l\in \mathbb{N}$ in order to maintain an even parity of channel pattern. For instance, we can devise a fixed block protocol such that a channel pattern is disjointly probed using $l$-pairs of TMSV states according to the disjoint probe-domain distribution {$\mc{S} = \{ \{12\}, \{34\}, \ldots,\{(2l-1)2\}\}$}. This corresponds to a $(k=1)$-LNN protocol, since each channel is probed in conjunction with precisely one of its neighbouring channels. \par

By exploiting the fidelity degeneracies of TMSV states and using Eq.~(\ref{eq:TM_EB}), there are some immediate expressions that can be found to derive analytical error bounds. 
Consider the task of CPF for $m$-length bosonic Gaussian channel patterns using this fixed, two-mode constrained block protocol. Assuming equal \textit{a priori} probabilities for all patterns $p_{\bs{i}} = 1/m, \forall \bs{i}\in \mc{U}_{\text{CPF}}$ it can be easily shown that the average error probability is bounded by
\begin{equation}
\begin{gathered}
\frac{1}{2m^2} D_{\bs{F}_2} [\mc{U}_{\text{CPF}} , 2M] \leq p_{\text{err}} \leq \frac{1}{m} D_{\bs{F}_2} [\mc{U}_{\text{CPF}} , M],\\
D_{\bs{F}_2} [\mc{U}_{\text{CPF}} ,M] = mF_{1:1}^M + {\big(2C_{m}^{2} - {m}\big)} F_{0:1}^{2M}.
\end{gathered}
\label{eq:CPF_disj}
\end{equation}
The derivation of this expression can be found in Appendix~\ref{sec:kLNNApp}.
For the purposes of position based quantum reading, these bounds can be shown to guarantee quantum advantage in some parameter ranges of background/target transmissivities and low signal energies. However if the number of channels is odd $m=2l+1$, there is a problem; one must utilise either a single mode input state on the remaining ``odd" channel, or employ a three-mode input state within the probe-domain distribution. Both of these scenarios lead to a rapid decay in discrimination performance. Meanwhile, for environment localisation tasks, the above bound reveals to be completely ineffective \cite{MP_IdlerFree}.\par

We can derive analytical error bounds with similar performance for the discrimination of uniform channel patterns (barcodes, $\bs{i}\in\mc{U}_{m}$) with an even number of channels,
\begin{equation}
\begin{gathered}
\frac{1}{2^{2m+1}} D_{\bs{F}_2} [\mc{U}_{m} ,2M] \leq p_{\text{err}} \leq \frac{1}{2^{m}} D_{\bs{F}_2} [\mc{U}_{m} ,M],\\
D_{\bs{F}_2} [\mc{U}_{m} ,\hspace{-0.5mm}M] = \left[1\hspace{-0.5mm} + \hspace{-0.5mm}F_{0:1}^M \hspace{-0.5mm}+ \hspace{-0.5mm}F_{1:2}^M \hspace{-0.5mm}+\hspace{-0.5mm}\frac{F_{1:1}^M \hspace{-0.5mm}+\hspace{-0.5mm} F_{0:2}^M}{2}\right]^{\frac{m}{2}}-1 .
\end{gathered}
\end{equation}
Again, these bounds offer quantum advantage for quantum reading, but fail to do so for environment localisation. In the following section, we show how dynamic discrimination protocols can remedy these issues.

\subsection{Channel Position Finding}
Let us focus on the use of $k$-LNN dynamic block protocols for the multi-channel discrimination task of CPF. Exact error bounds corresponding to these protocols can be derived using Eq.~(\ref{eq:TM_EB}), and admit a remarkably compact form. Indeed, it is sufficient to insert the modified image space $\{ \bs{\nu_i} \}_{\bs{i}\in\mc{U}_{\text{CPF}}}$ into Eq.~(\ref{eq:TM_EB}) and semi-numerically derive exact expressions. However, we show in Appendix \ref{sec:kLNNApp} how the following results can be derived in an intuitive fashion.

For a valid $k$-LNN dynamic protocol using $M$-copy input probes, and assuming equal priors we find the error probability of classification is bounded according to
\begin{gather}
\begin{gathered}
\frac{1}{2m^2} D_{\bs{F}_2}^{k}[2M] \leq p_{\text{err}} \leq \frac{1}{m} D_{\bs{F}_2}^{k} [M],\\
D_{\bs{F}_2}^{k} [M] = km(F_{0:1}^{2k-2} F_{1:1})^{M} + {(2C_{m}^{2} - km)} F_{0:1}^{2kM}.
\end{gathered}
\label{eq:kLNNCPF}
\end{gather}
Here we have dropped explicit notation dependence on the image space $D_{\bs{F}_2}^{k} [\mc{U}_{\text{CPF}} ,M] = D_{\bs{F}_2}^{k} [M]$ for convenience and clarity.
Clearly, when assuming an even number of channels, the fixed block protocol bounds from Eq.~(\ref{eq:CPF_disj}) are reproduced by setting $k=1$. Yet for valid values of $k>1$ the parity assumption no longer matters, and we gather a spectrum of CPF error bounds for any $m$ which utilise different collections of TMSV states.\par

The efficacy of these bounds can be readily studied by supplementing the exact expressions for the two-mode sub-fidelities $F_{u:v}$ into the above expressions (see Appendix \ref{sec:SubFids}).
For quantum reading purposes, these bounds observe discrimination performance close to idler-assistance, and guarantee strong quantum advantage for all $k$ and many channel parameters. In many cases, the performances $2\leq k\leq m-1$ are very similar, and it is sufficient/convenient to use $k=2$. However, in the setting of environment localisation the employment of overlapping probe-domains becomes critical to performance, and further increasing the neighbourhood parameter $k$ leads to huge performance gains. Indeed, disjoint probing structures in this regime fail dramatically, but are remedied by dynamic probing. \par

In both cases, we find that the best performance is achieved by the maximising the number of overlapping probe domains, i.e.~$k = k_{\max} = m-1$, such that each channel is probed in conjunction with every other channel. 
The performance of a $k_{\max}$-LNN dynamic protocol for CPF is characterised by the remarkably succinct bounds,
\begin{align}
p_{\text{err}} &\geq \frac{(m-1)}{2m}(F_{0:1}^{2(m-2)} F_{1:1})^{2M}, \label{eq:CPF_ifLB}\\
p_{\text{err}} &\leq (m-1)(F_{0:1}^{2(m-2)} F_{1:1})^{M}  \label{eq:CPF_ifUB}.
\end{align}
Quantum advantage can be guaranteed in this setting when the quantum upper bound in Eq.~(\ref{eq:CPF_ifUB}) is smaller than the best known classical lower bound \cite{OptEnvLoc, EntEnhanced}. For CPF, the $M$-copy, classical lower bound is given by
\begin{equation}
p_{\text{err}} \geq \frac{(m-1)}{2m} F_{\text{cl}}^{4M}.
\end{equation}
where $F_{\text{cl}}$ is the fidelity between the optimal, single mode coherent states used for classical multi-channel discrimination (see Appendix \ref{sec:ClassF} for explicit expressions). \par
In order to fairly compare these protocols, we must also ensure they both utilise the same average channel use, $\bar{M}$. The dynamic protocol has an average channel use of 
\begin{equation}
\bar{M}(\mc{S}_{k_{\max}}^m) = k_{\max} M = (m-1) M,
\end{equation}
since it probes every channel $(m-1)$ times per round of discrimination. Therefore to match resources with the classical protocol we provide it with only $M/(m-1)$ probe copies. This adjustment is made when comparing any quantum/classical protocols in this paper. Using the log-ratio of these bounds (see Eq.~(\ref{eq:log_ratio})) for a sensitive measure of quantum advantage, one finds that advantage is guaranteed when 
\begin{gather}
{M} \geq \frac{\log_{10}(2m)}{4\log_{10}\left({F_{\text{cl}}}\right) - \frac{1}{(m-1)}{\log_{10}\big({F_{0:1}^{2(m-2)} F_{1:1}}\big)}}  \label{eq:gadv}.
\end{gather}
This gives us a lower bound on the number of probe copies required to guarantee quantum advantage for idler-free CPF under certain environmental parameters.

\subsection{Larger Image Spaces and Generalised CPF}

It is desirable to generalise these results to larger, more complex image spaces, such as those pertaining to the task of $u$-CPF or barcode decoding. Yet for generally $k$-LNN dynamic protocols, this becomes very difficult. Analytical expansions of the total error quantity in Eq.~(\ref{eq:TM_EB}) do not necessarily take expedient forms for arbitrary image spaces $\mc{U}$ or all values of $k$. However, it is always possible to semi-numerically generate these expressions via symbolic programming, which can be stored and  used at will.

Fortunately, the output states of $k_{\max}$-LNN dynamic protocols regain the fidelity degeneracy properties of $m$-mode CV-GHZ states. More precisely, 
when all channels are identically probed in conjunction with every other channel using TMSV states, the fidelity of output states becomes degenerate with respect to Hamming distance preserving subspaces of $u/v$-CPF image spaces. This is thanks to the symmetry properties of the total input and output probe states.
Hence, it is possible to write a theorem analogous to Theorem \ref{theorem:Degen1} for the output states of $k_{\max}$-LNN protocols, rather than $m$-mode CV-GHZ states (see Theorem \ref{theorem:DegenMax}, Appendix \ref{sec:FidDegen} for more details).

Following this logic, it is now possible to write a new unique fidelity function between output states of $k_{\max}$-LNN protocols. Consider a pair of $m$-length channel patterns $\bs{i}$ and $\bs{i}^{\prime}$ which have Hamming distance $d_{h}(\bs{i},\bs{i}^{\prime})= d$ and each contain precisely $u$ target channels, i.e.~$\bs{i}, \bs{i}^{\prime} \in \mc{U}_{\text{CPF}}^u$. Then the unique fidelity function takes the form,
\begin{equation}
F_{[u:u|d]}^{k_{\max}} = F_{0:1}^{\frac{d[2(m-u) -d]}{2}}F_{1:1}^{\frac{d^2}{4}} F_{0:2}^{\frac{d(d-2)}{4}}F_{1:2}^{\frac{d(2u-d)}{2}}. \label{eq:kmaxUFid}
\end{equation}
Clearly, when $d=0$, $F_{[u:u|d]}^{k_{\max}} = 1$ as expected. Furthermore, it can be seen that the error bounds in Eqs.~(\ref{eq:CPF_ifLB}) and (\ref{eq:CPF_ifUB}) emerge from this expression, by setting $u=1$ and $d=2$ for the task of CPF and correctly normalising. For more general output fidelities between $\bs{i} \in \mc{U}_{\text{CPF}}^{u}$ and $\bs{i}^{\prime} \in \mc{U}_{\text{CPF}}^{v}$ channel patterns where $u < v$, the unique fidelity function is
\begin{equation}
{F}_{[u:v|d]}^{k_{\max}} \defeq
F_{[u:u|d]}^{k_{\max}} \Big[
{F_{0:1}F_{0:2}^{\frac{v-u}{2} } }{
F_{1:1}^{-\frac{v-u}{2} } F_{1:2}^{ (d-2(v-u))} }\Big]^{\frac{v-u}{2}} ,
\label{eq:kmaxUVFid}
\end{equation}
See Appendix~\ref{sec:GenCPF} for the explicit derivation of these functions.

With these output fidelity expressions in hand, it is straightforward to construct analytical error bounds for the discrimination of more complex image spaces via $k_{\max}$-LNN protocols. Using recently developed tools from Refs.~\cite{PatternRecog,ThermalPatt}, it is possible to simplify the sum over unequal channel patterns in the error bounds of Eqs.~(\ref{eq:LB}) and (\ref{eq:UB}). This is achieved by re-parameterising the sum in terms of the Hamming distance between equal or unequal $u/v$-CPF image spaces, and applying basic counting arguments. For explicit expressions of these bounds, we provide full details in Appendix~\ref{sec:GenCPF}.

\begin{figure}
\hspace{-0.8cm} Pure loss: (a) 1-CPF  \hspace{1.5cm} (b) 8-CPF \\
\includegraphics[width=\linewidth]{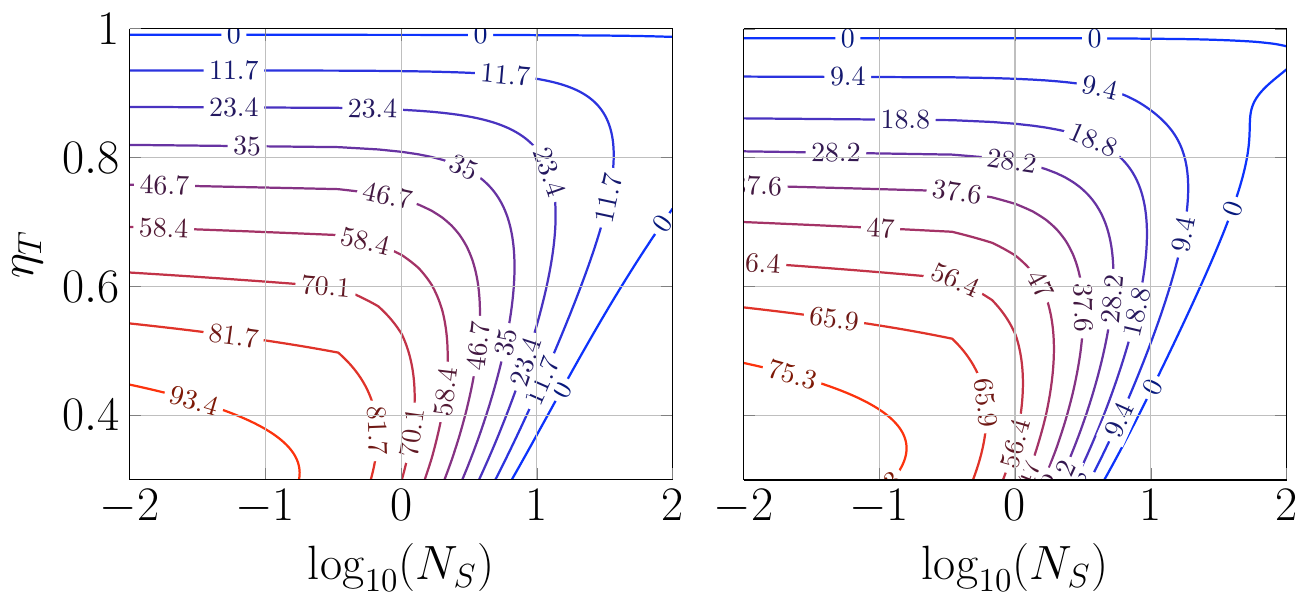}
\hspace{-0.7cm}Add noise: (c) 1-CPF \hspace{1.5cm} (d) 8-CPF \\
\includegraphics[width=\linewidth]{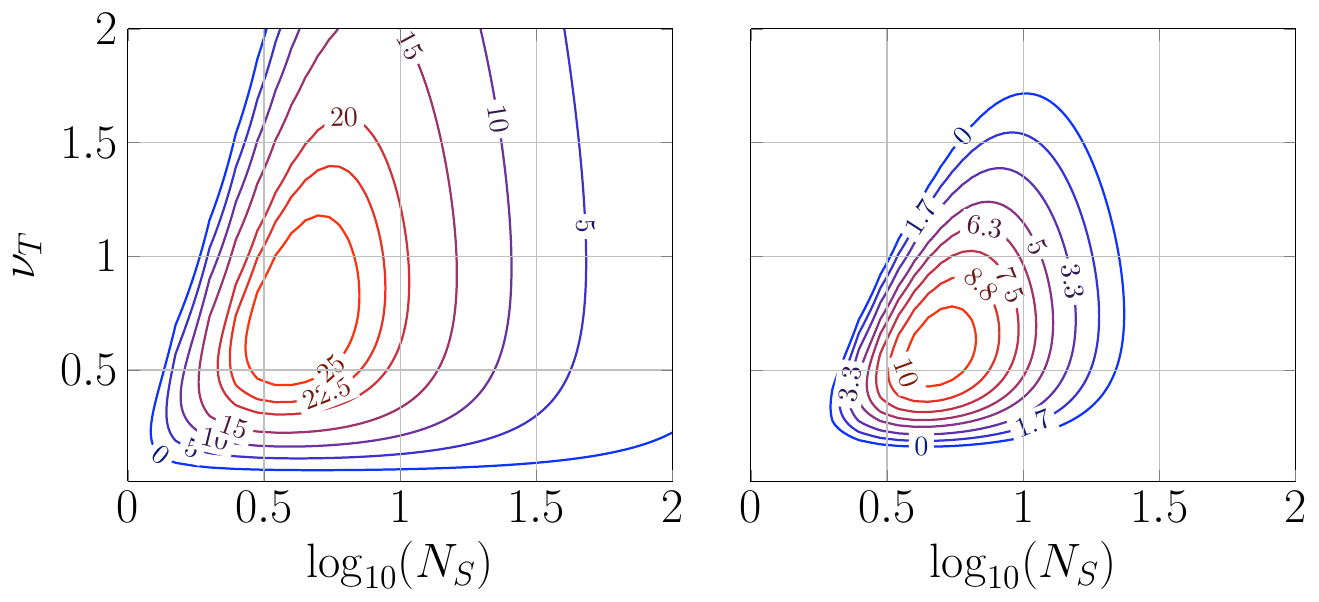}
\caption{Log-ratio of guaranteed advantage $\Delta_{\text{adv}}$ for 64 cell patterns: (a),(b) Position based quantum reading of pure-loss channels with ideal backgrounds $\eta_B = 1$, and a fixed total number of photons per channel $MN_S = 500$. (c),(d) Environment localisation, with background additive noise channels with $\nu_B = 0.02$. The total number of photons per channel is fixed at $MN_S = 1000$. The target channel parameters $\eta_T$ and $\nu_T$ are dimensionless, while $N_S$ is the mean photon number per signal mode.  }
\label{fig:CPFfigs}
\end{figure}

\begin{figure}
(a) Pure loss: $M=500$, $N_S = 10$.\\
\includegraphics[width=0.95\linewidth]{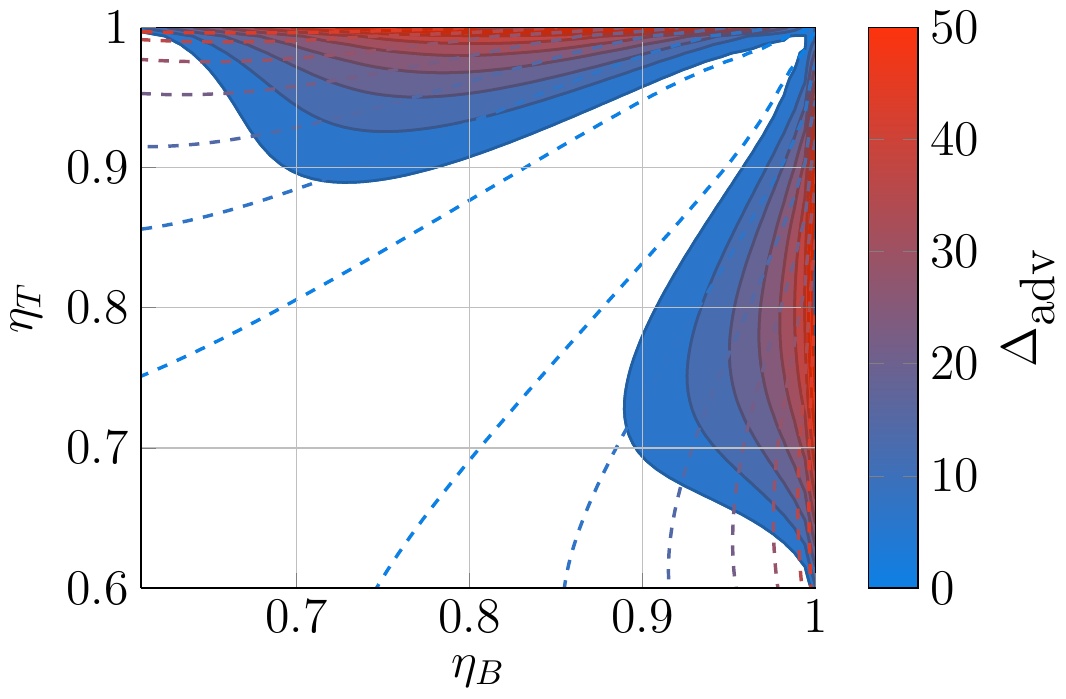}\\
(b) Add noise: $M=2000, N_S = 20$.\\
\includegraphics[width=0.95\linewidth]{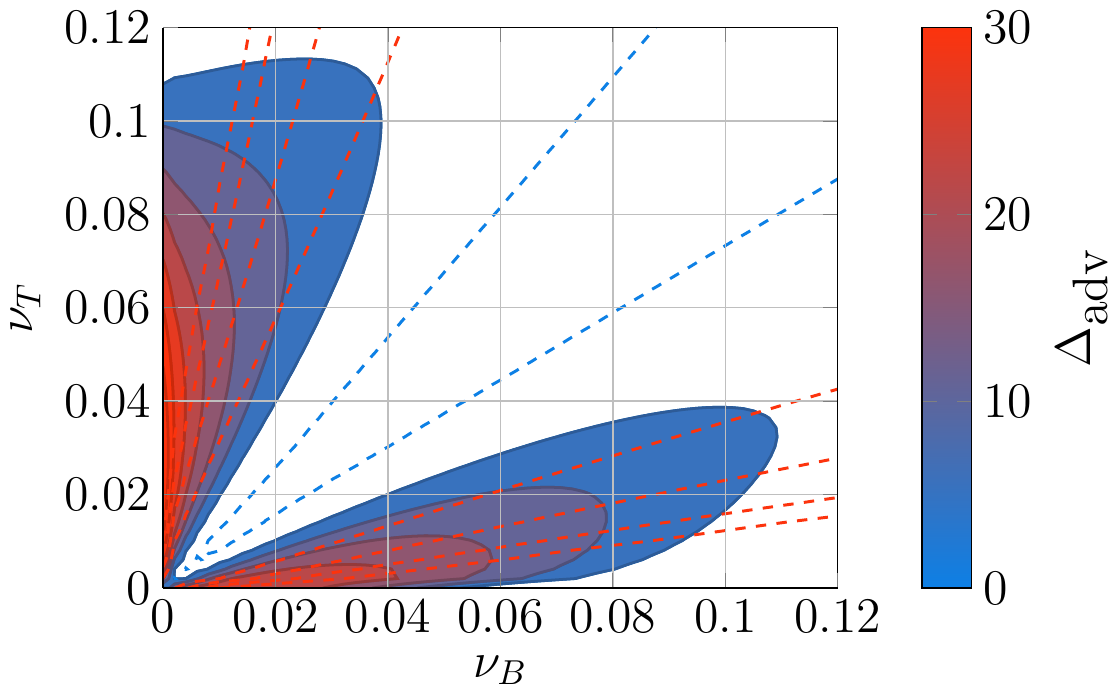}
\caption{Log-ratio of guaranteed advantage $\Delta_{\text{adv}}$ for the idler-free discrimination of $m=25$ cell uniform patterns (barcodes), for (a) quantum reading using $MN_S = 5000$ photons per channel, and (b) environment localisation using $MN_S = 40000$ photons per channel. Filled plots depict regions of advantage, while dashed lines depict the advantage obtained by idler-assisted protocols. All quantities plotted are dimensionless.}
\label{fig:Unifigs}
\end{figure}

\subsection{Quantum Advantage and Performance}
Using the bounds derived in the previous section, unassisted, dynamic protocols can be shown to significantly outperform optimal classical ones for the tasks of quantum reading, and environment localisation. The key metric of advantage used is the log-ratio of the optimal classical lower bound $p_{\text{err}}^{\text{cl,LB}}$, to the upper bound of our idler-free quantum protocols $p_{\text{err}}^{\text{q,UB}}$. Quantum advantage is guaranteed when,
\begin{equation}
\Delta_{\text{adv}} \defeq \log_{10}\left( \frac{p_{\text{err}}^{\text{cl,LB}}}{ p_{\text{err}}^{\text{q,UB}}}\right) \geq 0, \label{eq:log_ratio}
\end{equation}
granted that both protocols possess identical resources (carefully accounting for the increased average channel use of dynamic protocols).\par

\subsubsection{Quantum Reading}
Fig.~\ref{fig:CPFfigs} depicts regions of guaranteed quantum advantage for the tasks of (a) single CPF and (b) 8-CPF quantum reading of 64 cell patterns of pure-loss channels. In both cases, we consider an ideal background channel with transmissivity $\eta_B = 1$, and vary the target transmissivity, $\eta_T$. Here the total number of photons is fixed at $MN_S = 500$ per channel over the entire protocol. 
We find that in both scenarios, when the mean photon energy satisfies $0 < N_S \lesssim 100$ then the dynamic $k_{\max}$-LNN protocol exhibits strong quantum advantage in a vast region of parameter space. When employing very low energy probes, $0 < N_S \lesssim 5$, one observes a very high level of guaranteed advantage. Dynamic protocols offer a clear enhancement over alternative unassisted protocols, such as the fixed strategy in \cite{Jason_IdlerFree}. \par

More complex image spaces with variable target channel number can be investigated using Eq.~(\ref{eq:kBCPF_Bounds}). Fig.~\ref{fig:Unifigs} (a) plots regions of parameter space within which guaranteed advantage is possible for the quantum reading of uniformly distributed $m=25$ cell patterns (barcodes). When $N_S = 10$ and $M=500$, we may guarantee advantage provided either target or background transmissivity satisfy $(\eta_i \gtrsim 0.88) \wedge (\eta_j \gtrsim 0.6)$ for $i\neq j \in \{B,T\}$. Idler-assistance is also able to achieve advantage in lower transmissivity regimes, where unassisted protocols cannot. However, in regions where unassisted protocols are effective, they achieve close to identical performance with idler-assisted strategies.

\subsubsection{Environment Localisation}

Similar advantage is shown in Fig.~\ref{fig:CPFfigs} for (c) single CPF and (d) 8-CPF environment localisation of additive-noise channels for 64 cell patterns. In this setting, we fix the background noise at $\nu_B = 0.02$, and set the total photon input per channel at $MN_S = 1000$. Entanglement shared between probe states is much more fragile with respect to noise than loss, nonetheless, significant levels of guaranteed advantage are still attainable for both tasks. Yet, the parameter space within which advantage can be obtained shrinks with respect to increasing image space complexity.\par

For environment localisation, things are more difficult when the target channel number is variable. The fragility of entanglement to high levels of noise means that unassisted protocols are only effective in a low-noise regime. In CPF settings with large numbers of background channels of low noise (such as in Fig.~\ref{fig:CPFfigs}(c) and (d)) the input states remain resilient, and entanglement can be exploited effectively for enhanced performance. 
Fig.~\ref{fig:Unifigs}(b) plots regions of parameter space within which guaranteed advantage is obtained for environment localisation of $m=25$ cell barcodes. Here, the quantum enhancement becomes much more difficult to maintain. Given $N_S = 20$ and $M=2000$, advantage is approximately guaranteed when $(\nu_i \lesssim 0.04) \wedge (\nu_j \lesssim 0.11)$ for $i\neq j \in \{B,T\}$. Since additive-noise channels describe an idealised thermal-loss channel, we can conclude that unassisted, advantageous protocols for environment localisation are only feasible in a very low loss, low noise setting.

\section{Conclusion and Outlook\label{sec:Concl}}
We have derived exact, analytical error bounds for a class of idler-free multi-channel discrimination strategies which we label dynamic protocols. These unassisted, non-adaptive protocols exploit entanglement without ancillary systems, and use variable probing configurations to minimise their average error probability over binary image spaces. Applying these protocols in the bosonic setting for the tasks of quantum reading and environment localisation, we have shown that significant quantum advantage can be obtained over the best known classical strategies in a variety of settings, without the need for idler-assistance.\par
These dynamic discrimination methods offer new, experimentally feasible ways to exhibit quantum superiority over classical methods without the requirement of quantum memories. In particular, their efficiency over pure-loss channels make them ideal for realising high-rate, position-based quantum reading over large collections of cells in an optical memory.\par
There are many future directions of interest regarding these protocols. The extension to $(d>2)$-ary channel patterns and more complex image spaces offers a critical area of interest. Furthermore, combining idler-free strategies with with cutting edge machine-learning methods could offer a practical path for quantum-enhanced pattern recognition in an optical setting.

\begin{acknowledgements}
C.H acknowledges funding from the EPSRC via a Doctoral Training Partnership (EP/R513386/1). S.P acknowledges funding from the European Union’s Horizon 2020 Research and Innovation Action under grant agreement No.~862644 (Quantum readout techniques and technologies, QUARTET).
\end{acknowledgements}

%\bibliography{dynamic_mcd}


\appendix

%\begin{widetext}
%\begin{widetext}
\section{Fidelity Degeneracies\label{sec:FidDegen}}
In order to compute the error bound quantities discussed in Section \ref{sec:Methods}, we require the following tools:

\subsection{CV-GHZ Output States}
\begin{lemma}\label{lemma:GUS}
For $m$-mode CV-GHZ states interacting with strictly $u$-CPF, GPI channel patterns $\bs{i}\in\mc{U}_{\text{\emph{CPF}}}^u$, the output ensemble possesses geometrical uniform symmetry (GUS) assuming equal priors $p_{\bs{i}} = {1}/C_{m}^u$.
\end{lemma}
\begin{proof}
The strictly $u$-CPF image space with equal priors is given by $\mc{U}_{\text{CPF}}^u$ such that all $\bs{i}$ occur with probability $p_{\bs{i}} = {1}/{C_m^u}$ and possess exactly $u$-target channels. The output ensemble from CV-GHZ state probing provides an output ensemble $\{{1}/{C_m^u}; \Phi_{\bs{i},\mu}\}_{\bs{i}\in\mc{U}_{\text{CPF}}^u}$. Since the input state $\Phi_{\mu}$ is fully symmetric and the number of target channels is fixed, then any possible output state can be generated by permuting the modes of an initial output state. That is, we can devise a set of a symmetry unitaries $\{ S_{\bs{i}}^u \}_{\bs{i}\in\mc{U}_{\text{CPF}}^u}$ such that for all $\bs{i},\bs{j}  \in \mc{U}_{\text{{CPF}}}^u$ we can write $\Phi_{\bs{i},\mu} = S_{\bs{i}}^u \Phi_{\bs{j},\mu} S_{\bs{i}}^{u\dagger}$. When $\bs{i}=\bs{j}$, then the permutation operator is simply the $m$-mode identity. 
\end{proof}\\

Hence all output states from a $u$-CPF image space have GUS, and this can be exploited with respect to the fidelity.\\

\begin{proof} \textit{(Proof of Theorem \ref{theorem:Degen1}):}\\
Consider the output quantum state ensembles according to the probing of a $u$-CPF image space $\bs{i}\in \mc{U}_{\text{{CPF}}}^u$ and the $v$-CPF image space $\bs{i}^{\prime} \in \mc{U}_{\text{{CPF}}}^v$ using bosonic CV-GHZ states $\Phi_{\mu}$. By Lemma~\ref{lemma:GUS} these output ensembles possess GUS. Therefore we can generate any possible pair of output states using the permutation operators $S_{\bs{i}}^u$ and $S_{{\bs{i}^{\prime}}}^v$ respectively, and using any initial output state from each space, which we label $\Phi_{\bs{j},\mu}$ and $\Phi_{\bs{j}^{\prime},\mu}$. Then any fidelity between output states from these image spaces is given by,
\begin{equation}
F( \Phi_{\bs{i},\mu}, \Phi_{\bs{i}^{\prime},\mu}) = F(S_{\bs{i}}^u \Phi_{\bs{j},\mu} {S_{\bs{i}}^{u\dagger}}, S_{\bs{i}^{\prime}}^v \Phi_{\bs{j}^{\prime},\mu} {S_{\bs{i}^{\prime}}^{v\dagger}}).
\end{equation}
The fidelity will only be degenerate when the symmetry unitaries $S_{\bs{i}}^{u} = S_{\bs{i}^{\prime}}^{v}$ are equal, since identically applied unitaries $U$ cannot increase the distance between quantum states $F(U\rho U^{\dagger}, U\sigma U^{\dagger}) = F(\rho,\sigma)$. These symmetry operators $S_{\bs{i}}^{u} = S_{\bs{i}^{\prime}}^{v} = P_{\bs{a}\bs{b}}$ are simply permutation unitaries that exchange the modes $\bs{a} = \{a_1,\ldots\,a_n\}$ and $\bs{b} = \{b_1,\ldots,b_n\}$ \footnote{$P_{\bs{a}\bs{b}}$ are clearly unitary $P_{\bs{a}\bs{b}}^{\dagger}P_{\bs{a}\bs{b}} = P_{\bs{a}\bs{b}} P_{\bs{a}\bs{b}}^{\dagger} = I$, since permuting/re-permuting produces a net zero transformation}. Since the identical permutation of modes of both states does not alter the Hamming distance between re-ordered patterns, it follows that
\begin{equation}
F(P_{\bs{a}\bs{b}} \Phi_{\bs{i},\mu}P_{\bs{a}\bs{b}}^{\dagger} , P_{\bs{a}\bs{b}} \Phi_{\bs{i}^{\prime},\mu} P_{\bs{a}\bs{b}}^{\dagger}) = F(\Phi_{ \bs{i},\mu},\Phi_{\bs{i}^{\prime},\mu}),
\end{equation}
and therefore the fidelity is then degenerate with respect to sub-spaces $\mc{X} \subseteq \mc{U}_{\text{{CPF}}}^u$ and $\mc{Y} \subseteq \mc{U}_{\text{{CPF}}}^v$ within which the Hamming distance is preserved.
\end{proof}\\

The exploitation of fidelity degeneracies is the most critical tool in this paper, and is effectively used to simplify and analyse error bounds for general multipartite input states.

\subsection{$k_{\text{max}}$-LNN Output States}
Consider strictly $u$-CPF patterns $\bs{i} \in \kCPF{u}$. The use of a $k_{\max}$-LNN dynamic block protocol means that within every full round of discrimination, we will use TMSV states to probe every possible pair of channels within the channel pattern. Thus, each channel is probed in conjunction with exactly $k_{\max} = m - 1$ other channels. This probe domain distribution takes the general form,
\begin{equation}
 \mc{S}_{k_{\max}}^m =  \bigcup_{i=1}^{m-1} \bigcup_{j=i+1}^{m} \big\{ \{ i, j \} \big\}
\end{equation}
For example, for $m=4$, this probe-domain distribution takes the explicit form,
\begin{equation}
 \mc{S}_{3}^4 = \big\{\{1,2\}, \{1,3\}, \{1,4\}, \{2,3\}, \{2,4\}, \{3,4\}\big\}
\end{equation}
where it is clear that each channel label occurs precisely $k_{\max}=m-1=3$ times throughout the distribution.

A $k_{\text{max}}$-LNN dynamic protocol can be studied by performing the dynamic to fixed protocol transformation introduced in Ref.~\cite{MP_IdlerFree}; converting $m$-length channel patterns into a modified $m(m-1)$-length channel pattern within which copies of channels have been placed to mimic the dynamic protocol. The transformation of a generic $m$-length channel pattern then follows the probe-domain distribution as discussed in the main text,
 \begin{align}
  \bs{i} = \{i_1,\ldots,i_m\} \mapsto \bs{\nu_i} &  =\biguplus_{\bs{s}\in\mc{S}_{k_{\text{max}}}^m} \{ i_j \}_{j\in\bs{s}}, \\
 & = \biguplus_{n = 1}^{\frac{m(m-1)}{2}} \{i_j\}_{\bs{s}_n},
 \end{align}
where in the last expression we are explicitly indexing the probe domains within the total $k_{\max}$-LNN probe-domain distribution. Hence, there are exactly ${m(m-1)}/{2}$ channel pairs that are considered within $\mc{S}_{k_{\max}}^m$.

Given this probe domain distribution, we may define the corresponding unassisted quantum input state that it constructs. Let us denote TMSV states using $\varphi_{\mu}$ so to distinguish them more easily from potentially larger CV-GHZ states. Then the (single-copy, $M=1$) input state for a $k_{\max}$-LNN dynamic protocol input state as the tensor product $\frac{m(m-1)}{2}$ TMSV states, accounting for each pair of channels that are probed together in the protocol,
\begin{equation}
\bs{\varphi}_{\mu} \defeq \bigotimes_{c=1}^{\frac{m(m-1)}{2}} \varphi_{\mu}.
\end{equation}
Since each sub-state in the global input state is of the same mean photon number and are themselves fully symmetric, then the total state $\bs{\varphi}_{\mu}$ is also fully symmetric. 
We can then analyse the performance of these protocols by studying the output states of $\bs{\varphi}_{\mu}$ being irradiated over the set of $m(m-1)$-length modified channel patterns $\{\bs{\nu_i}\}_{\bs{i}\in\mc{U}}$.
Any output state then takes the form
\begin{equation}
\bs{\varphi}_{\bs{\nu_i}, \mu} = \mc{E}_{\bs{\nu_i}} \left( \bs{\varphi}_{\mu} \right) = \bigotimes_{\bs{s}\in\mc{S}_{k_{\text{max}}}^m} \mc{E}_{\bs{\nu_{i_{s}}}}(\varphi_{\mu}),
\end{equation}
where $\bs{\nu_{i_{s}}}$ is a sub-pattern from the larger modified channel pattern. 

We now reach an extremely important point. Let's explicitly consider the transformation of a $u$-CPF image space according to $\mc{S}_{k_{\max}}^m$, i.e.~$ \mc{U}_{\text{CPF}}^u \mapsto \{ \bs{\nu_i} \}_{\bs{i} \in\mc{U}_{\text{CPF}}^u} $. For every possible $\bs{i}\in\mc{U}_{\text{CPF}}^u$ there are precisely $u$ target channels. As discussed in Lemma \ref{lemma:GUS}, since the number of targets is always $u$ then every pattern $\bs{i}$ can be generated by simply permuting the positions of target and background channels within the pattern. Under the $k_{\text{max}}$-LNN dynamic protocol every channel is probed in conjunction with every other channel, therefore every pattern in the modified image space must also contain precisely $(m-1)u$ target channels. It then follows that any pattern $\bs{\nu_i} \in \{ \bs{\nu_i} \}_{\bs{i} \in\mc{U}_{\text{CPF}}^u}$ can be generated by permuting the positions of \textit{pairs of target and background channels} within the pattern. Because of this, we find that the output ensembles of these protocols interacting with $u$-CPF image spaces possess GUS.\\

\begin{lemma}\label{lemma:GUSMax}
The output ensemble of a $k_{\text{max}}$-LNN dynamic protocol interacting with strictly $u$-CPF, GPI channel patterns $\bs{i} \in \mc{U}_{\text{CPF}}^u$ possesses GUS, assuming equal priors $p_{\bs{i}} = 1/C_m^u$.
\end{lemma}
\begin{proof}
Under this protocol and image space, an output ensemble then takes the form, $\{1/C_m^u; \bs{\varphi}_{\bs{\nu_i},\mu}\}_{\bs{i}\in \mc{U}_{\text{CPF}}^u}$. Since $\bs{\varphi}_{\mu}$ is fully symmetric and the number of target channels is fixed for all patterns in the modified pattern space $\{ \bs{\nu_i} \}_{\bs{i} \in\mc{U}_{\text{CPF}}^u}$, then any possible output state can be generated by permuting/rearranging TMSV sub-states $\mc{E}_{\bs{\nu_{i_{s}}}}(\varphi_{\mu})$ within the global output state. Therefore, we can devise a set of symmetry unitaries $\{ S_{\bs{\nu_i}}^u \}_{\bs{i} \in\mc{U}_{\text{CPF}}^u}$ such that for all $\bs{i},\bs{j}\in\mc{U}_{\text{CPF}}^u$ we can write
$
\bs{\varphi}_{\bs{\nu_i},\mu} = S_{\bs{\nu_i}}^u \bs{\varphi}_{\bs{\nu_j},\mu} S_{\bs{\nu_i}}^{u \dag}$ which permute/exchange pairs of modes within the output state. When $\bs{i} = \bs{j}$, then  $\bs{\nu_i} = \bs{\nu_j}$ and the symmetry unitaries are simply $m(m-1)$-mode identity operators. Hence the output ensemble possesses GUS.
\end{proof}\\

Interestingly, GUS does not apply to output ensembles from dynamic protocols generated from any other $\mc{S}_k^m$ probe distributions. This is because no other $k$-LNN dynamic protocol invokes a channel pattern transformation which ensures that all patterns in the modified image space possess the same number of targets. By choosing any $k < m-1$, then the corresponding modified image space may have varying numbers of target channels in each $\bs{\nu_i}$. When this is the case, it is not possible to transform between patterns via symmetry operators, since symmetry operators are unable to change the number of targets in the channel pattern; thus removing GUS.
As a result of Lemma~\ref{lemma:GUSMax}, we can produce an analogous fidelity degeneracy theorem for $k_{\max}$-LNN dynamic block protocols.

\begin{theorem}\label{theorem:DegenMax}
Consider two subsets of $u/v$-CPF image spaces respectively $\mc{X}$ and $\mc{Y}$, where
\begin{gather} 
d_h(\bs{i},\bs{i}^{\prime}) = d > 0,\\
\forall \bs{i} \in \mc{X}\subseteq\mc{U}_{\text{\emph{CPF}}}^u, \bs{i}^{\prime} \in \mc{Y} \subseteq\mc{U}_{\text{\emph{CPF}}}^v.
\end{gather} 
This means that $\mc{X}$ and $\mc{Y}$ are Hamming distance preserving subsets of $~\mc{U}_{\text{\emph{CPF}}}^u$ and $\mc{U}_{\text{\emph{CPF}}}^v$. 
Let $\bs{i},\bs{i}^{\prime} \mapsto \bs{\nu_i},\bs{\nu}_{\bs{i}^{\prime}}$ be the mapping of the original patterns to a modified image space according to a $k_{\text{max}}$-LNN dynamic block protocol.
For two GPI channel patterns $\mc{E}_{\bs{i}}, \mc{E}_{\bs{i}^{\prime}}$ with identical physical properties, and an $m(m-1)$-mode product of TMSV states $\bs{\varphi}_{\mu}$,
it follows that the fidelity
\begin{equation}
F\left[\mc{E}_{\bs{\nu_i}}(\bs{\varphi}_{\mu}),\mc{E}_{\bs{\nu}_{\bs{i}^{\prime}}}(\bs{\varphi}_{\mu}) \right] = F(\bs{\varphi}_{\bs{\nu_i},\mu},\bs{\varphi}_{\bs{\nu}_{\bs{i}^{\prime}},\mu} ),
\end{equation}
is degenerate for all $\mc{X} \ni \bs{i} \neq \bs{i}^{\prime} \in \mc{Y}$. 
\end{theorem}

\begin{proof}
By Lemma \ref{lemma:GUSMax} the output ensemble satisfies GUS, thus the proof follows in direct analogy to Theorem \ref{theorem:Degen1}. 
\end{proof}\\

This theorem is used in Appendix \ref{sec:GenCPF} to derive explicit fidelity formulae for Hamming distance preserving sub-spaces of general CPF image spaces.\par

\section{$k$-LNN CPF Bounds \label{sec:kLNNApp}}
While the CPF bounds presented in Eq.~(\ref{eq:kLNNCPF}) can be immediately generated by means of semi-numerical methods and Eq.~(\ref{eq:TM_EB}), it is possible to provide a more intuitive explanation behind these results.\par
\subsection{$k=1$ (Disjoint)}
Consider first the case of $k=1$, and an even number of channels, $m=2l > 2, l\in\mathbb{N}$. Then the probe distribution is disjoint and requires no dynamic to fixed protocol image space modification. 
Since the task of CPF is embodied by an image space with patterns that have precisely one target channel, any pair of CPF patterns can be characterised by the pair of indices of those single target channels, i.e $\bs{i},\bs{i}^{\prime} \rightarrow i,i^{\prime}$. For example, the 3-length pattern $\bs{i} = \{B,T,B\} = \{0,1,0\} \rightarrow i = 2$. 
When constrained to two mode input states over an even number of $m$-channels and iterating over all unequal channel patterns, then there are only two unique sub-fidelities that will occur:
\begin{enumerate}
\item A pair of sub-patterns $i_{\bs{s}} \in \{1,2\} \rightarrow \bs{i}_{\bs{s}}\in \{\{1,0\},\{0,1\}\} $ and $ i_{\bs{s}}^{\prime} =0 \rightarrow \bs{i}_{\bs{s}}^{\prime} = \{0,0\}$, leading to the sub-fidelity $F_{0:1}$.
\item A pair of sub-patterns $i_{\bs{s}} = 1 \mapsto \bs{i}_{\bs{s}} = \{1,0\}$ and $i_{\bs{s}}^{\prime} = 2 \mapsto \bs{i}_{\bs{s}}^{\prime} = \{0,1\}$ (and vice versa) leading to the sub-fidelity $F_{1:1}$.
\end{enumerate}
What is the distribution of these fidelities when we consider the sum over all non-unit fidelities throughout the image space $\mc{U}_{\text{CPF}}$? The sub-fidelity $F_{1:1}$ will only occur when we sample a sub-pattern pair with targets that belong to the same probe-domain, $i_{\bs{s}}^{\prime} \in \mc{N}_1(i_{\bs{s}})$. This will occur exactly $m$ times, since each channel only has one neighbour.\par 
For all other sampled sub-pattern pairs, the sub-fidelity $F_{0:1}$ will occur twice per pattern when the target channels are outside of each other's neighbourhoods. This remaining number of patterns is $(2C_{m}^2 - m)$, which leads to the following fidelity distribution:
\begin{equation}
m F_{1:1} + (2C_{m}^2 - m)F_{0:1}^2,
\end{equation}
which can then be used to compose Eq.~(\ref{eq:CPF_disj}). When the number of channels is no longer even, this result does not apply since we cannot disjointly partition an odd number of channels into two-mode groupings.

\subsection{$k \geq 2$ (Non-Disjoint)}
Let us complicate matters and let $k=2$. Now there is no restriction on $m$-channel parity, since we can construct lattices of 2-local neighbourhoods regardless of the number of channels. Defining the non-disjoint probe distribution $\mc{S}_{m,2}$, we transform the image space according to the dynamic to fixed protocol mapping, retrieving the modified image space $\{ \bs{\nu_i} \}_{\bs{i}\in\mc{U}_{\text{CPF}}}$. Disjointly interacting with this modified space using a $2m$-mode collection of TMSV states is then equivalent to studying the dynamic protocol. The extended neighbourhood domain alters the distribution of sub-fidelities across the pattern. Instead of having only one target channel per pattern $\bs{i}$, we have two targets per modified pattern $\bs{\nu_i}$. We may obtain the following scenarios:
\begin{enumerate}
\item Suppose we draw two unequal global patterns $i,i^{\prime}$. If $i^{\prime} \in \mc{N}_2(i)$ (i.e.~the target channel in $\bs{i}^{\prime}$ is in the neighbourhood of the target channel in $\bs{i}$) then we will obtain a fidelity $F_{1:1}$, for the same reasons as in the $k=1$ case. But now due to the copies of channels which are inserted into the modified image space, there will simultaneously be two further instances of the sub-pattern pairs $\bs{i}_{\bs{s}} = \{0,1\}$ and $\bs{i}_{\bs{s}}^{\prime} = \{1,0\}$ (and vice versa). This unique sub-fidelity corresponds to $F_{0:1}^2 F_{1:1}$.

\item  Suppose that the target channel in $\bs{i}^{\prime}$ is not in the channel neighbourhood of $\bs{i}$, that is $i^{\prime} \notin \mc{N}_2(i)$. Then there will occur 4 instances of the sub-pattern pair $\bs{\nu}_{\bs{i}_{\bs{s}}} = \{0,0\}$ and $\bs{\nu}_{\bs{i}_{\bs{s}}^{\prime}} = \{1,0\}$ in the modified pair of patterns, since there are no common neighbours. This corresponds to the fidelity $F_{0:1}^{4}$.
\end{enumerate}

The distribution of these fidelities is easy: Each channel possesses $2$ neighbours, therefore the fidelity $F_{0:1}^2 F_{1:1}$ will occur exactly $2m$ times. The remaining sub-fidelities which take the form $F_{0:1}^{4}$ will occur $(2C_{m}^2 - 2m)$ times. The total distribution is thus,
\begin{equation}
2mF_{0:1}^2 F_{1:1} + 2(C_{m}^2 - m)F_{0:1}^4.
\end{equation}\par
The generalisation to $k$-neighbours is now immediate. For CPF image spaces, there will still only be two forms of output fidelity based on whether the target channel in the sampled pattern is within the other's $k$-width neighbourhood. As the neighbourhoods become larger, we inherit one fidelity more frequently than the other. If we utilise $k$-width neighbourhoods,
\begin{equation}
\prod_{F_{u:v}\in\bs{F}_2} F_{u:v}^{c_{u:v}(\bs{i},\bs{i}^{\prime})} = 
\begin{cases}
F_{0:1}^{2k - 2} F_{1:1}, & \text{if } i^{\prime}\in \mc{N}_k(i),\\
F_{0:1}^{2k}, & \text{if } i^{\prime}\notin \mc{N}_k(i).
\end{cases}
\end{equation}
This is due to the increase in copy channels across the modified channel pattern, reducing the number of instances of identical sub-patterns. The fidelity $F_{0:1}^{2k - 2} F_{1:1}$ will occur $km$ times since each channel has $k$-neighbours, and the remaining $(2C_m^2 - km)$ fidelities will be $F_{0:1}^{2k}$. The distribution then follows,
\begin{equation}
kmF_{0:1}^{(2k-2)} F_{1:1} + (2C_{m}^2 - km)F_{0:1}^{2k}.
\end{equation}
which produces Eq.~(\ref{eq:kLNNCPF}). It is then clear that when we set $k = k_{\max} = m-1$ that we gather
\begin{equation}
m(m-1)F_{0:1}^{2(m-2)} F_{1:1},
\end{equation}
where the sub-fidelity term $F_{0:1}^{2k}$ no longer has any contribution, since there is never an instance where the target channel index in $\bs{i}^{\prime}$ is not in the neighbourhood of the target channel in $\bs{i}$.

\begin{figure*}
\includegraphics[width=\linewidth]{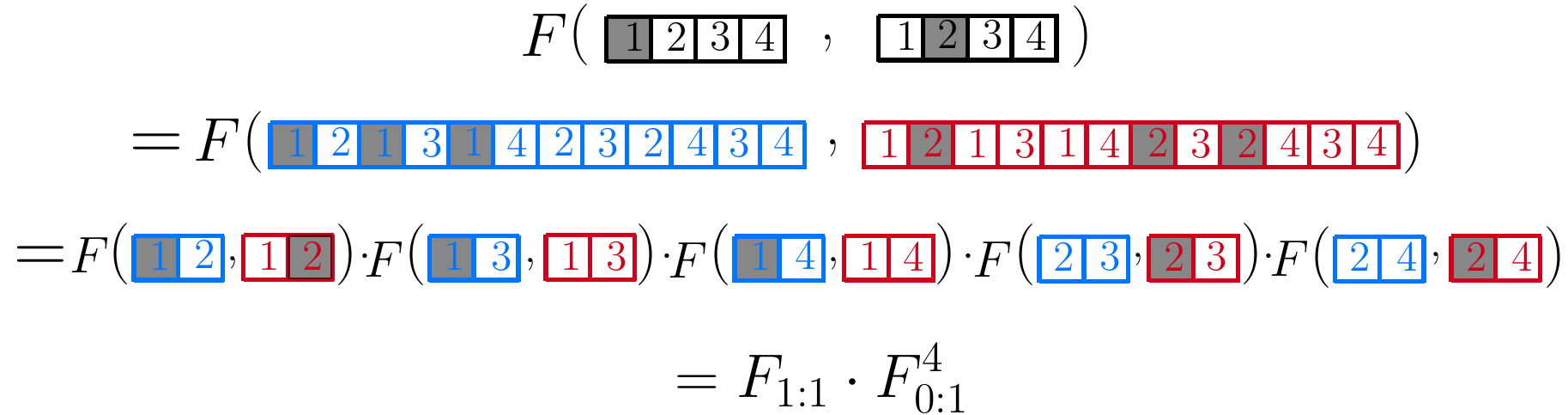}
\caption{An example simplification of the output fidelity between $m=4$ length CPF channel patterns when using a $k_{\max}$-LNN dynamic block protocol with probe-domain distribution {$\mc{S}_{3}^4 = \{\{1,2\},\{1,3\},\{1,4\},\{2,3\},\{2,4\},\{3,4\}\}$}. The fidelity between output states from modified channel patterns (blue and red) can be simplified via the multiplicativity of the fidelity and the degeneracy properties of CV-GHZ states. }
\end{figure*}

\section{$k_{\max}$-LNN Generalised CPF Bounds \label{sec:GenCPF}}

In the previous Appendix we derived the distributions of output fidelities according to $k$-LNN dynamic protocols over the CPF image space. More generally, we may investigate $u$-CPF image spaces in which each channel pattern $\bs{u}\in\mc{U}_{\text{CPF}}^u$ contains exactly $u$ target channels. For general $k$-width neighbourhoods this is analytically very difficult, however one can always algorithmically employ Eq.~(\ref{eq:TM_EB}) in order to generate analytical expressions. But when $k=k_{\max}=m-1$, the fidelity degeneracy properties unveiled in Theorem \ref{theorem:DegenMax} make it much easier to derive precise, analytical expressions. In particular, we wish to derive a formula for the output fidelity of unequal channel patterns from the same $u$-CPF image space, and then unique $u\neq v$-CPF image spaces.

As proven in Theorem \ref{theorem:DegenMax}, the output fidelities of $k_{\text{max}}$-LNN probing protocols defined by non-disjoint probe distributions $\mc{S}_{k_{\text{max}}}^m$ are degenerate with respect to $u/v$-CPF image sub-spaces that preserve the Hamming distance. This is extremely useful, since it allows us to utilise the tools from Section \ref{sec:GenMeth}, and in particular, make use of Eq.~(\ref{eq:TM_EB}).
Theorem 2 tells us that we can define a unique output fidelity for $k_{\text{max}}$-LNN dynamic protocols which is analogous to that of $m$-mode CV-GHZ states irradiated over $m$-length bosonic Gaussian channel patterns. 
More precisely, for a pair of {target numbers} $u$ and $v$ which refer to the number of target channels in $\bs{i}$ and $\bs{i}^{\prime}$ respectively and the Hamming distance between these patterns, ${d_h}(\bs{i},\bs{i}^{\prime}) = d$ we can define the unique fidelity function,
\begin{gather}
F_{[u:v|d]}^{k_{\max}} \defeq F({\bs{\varphi}}_{\bs{\nu_i},\mu},\bs{\varphi}_{\bs{\nu}_{\bs{i}^{\prime}},\mu}),\\
\forall \bs{i}\in\mc{U}_{\text{CPF}}^u, \text{  } \bs{i}^{\prime}\in\mc{U}_{\text{CPF}}^v, \text{ such that } {d_h}(\bs{i},\bs{i}^{\prime}) = d.
\end{gather}
Importantly, since the fidelity is multiplicative, it is possible to express any unique fidelity $F_{[u:v|d]}^{k_{\max}}$ as a product of sub-fidelities from the two-mode constrained unique fidelity set. More precisely, we can write
\begin{align}
F_{[u:v|d]}^{k_{\max}} &=  F(\bs{\varphi}_{\bs{\nu_{i}},\mu}, \bs{\varphi}_{\bs{\nu}_{\bs{i}^{\prime}},\mu}), \\
&= F\Big( \bigotimes_{\bs{s}} \bs{\varphi}_{\bs{\nu_{i_{s}}},\mu}, \bigotimes_{\bs{s}} \bs{\varphi}_{\bs{\nu}_{\bs{i_s}^{\prime}},\mu} \Big),\\
&\overset{(1)}{=} \prod_{\bs{s}\in\mc{S}_{k_{\text{max}}}^m}\hspace{-3mm} F (\varphi_{{\bs{\nu_{i_{s}}}},\mu}, \varphi_{\bs{\nu}_{\bs{i_s}^{\prime}}, \mu}),\\
&\overset{(2)}{=} \prod_{F_{x:y}\in \bs{F}_2} F_{x:y}^{f_{x:y}(u,v,d)}, \\
&\overset{(3)}{=} F_{0:1}^{f_{0:1}(u,v,d)}F_{1:1}^{f_{1:1}(u,v,d)}F_{0:2}^{f_{0:2}(u,v,d)}F_{1:2}^{f_{1:2}(u,v,d)}.
\end{align}
where in (1) we exploit the multiplicativity of the fidelity and in (2) we use the fidelity degeneracy of $k_{\max}$-LNN output states, and (3) we use the definition of the two-mode constrained unique fidelity set $\bs{F}_2$. Here, $f_{x:y}(u,v,d)$ is a counting function that counts the number of times the two-mode sub-fidelity $F_{x:y}$ contributes to the total output fidelity, given the target numbers $u,v$ and Hamming distance $d$ between two channel patterns $\bs{i},\bs{i}^{\prime}$.

\subsection{$u$-CPF Output Fidelities}
 Let us first focus on this unique fidelity function between strictly $u$-CPF patterns $\bs{i},\bs{i}^{\prime} \in \mc{U}_{\text{CPF}}^u$, i.e.~we wish to study the fidelity function $F_{[u:u|d]}^{k_{\max}}$. 
 We find that it is possible to explicitly derive the counting functions $f_{x:y}(u,v,d)$ via recursion. 
 Here we provide a brief sketch of how this can be achieved.
 Consider two initially identical $m=5$ length channel patterns $\bs{i} = \bs{i}^{\prime}$ for $u=2$,
 \begin{equation}
\bs{i} = \{1,0,0,1,0\}, ~ \bs{i}^{\prime} = \{1,0,0,1,0\}, ~d = 0,
\end{equation} 
where the Hamming distance $d$ between them is clearly zero since they are identical. The fidelity between these two identical output states is of course $F(\bs{\varphi}_{\bs{\nu_i},\mu},\bs{\varphi}_{\bs{\nu}_{\bs{i}^{\prime}},\mu}) = 1$ meaning that initially all the counting functions should satisfy $f_{x:y}(u,u,0) = 0$. What happens when the Hamming distance is increased?
The smallest increment that the Hamming distance can be increased by is 2, corresponding to the \textit{exchange} of the position of a target channel within one of the patterns. For example, we can exchange the positions of a target and background channel in $\bs{i}^{\prime}$ resulting in,
\begin{equation}
\bs{i} = \{1,0,0,1,0\}, ~ \bs{i}^{\prime} = \{1,0,0,0,1\}, ~d= 2.
\end{equation} 
A further target channel exchange in $\bs{i}^{\prime}$ can further increase the Hamming distance to its maximum value,
\begin{equation}
\bs{i} = \{1,0,0,1,0\}, ~ \bs{i}^{\prime} = \{0,1,0,0,1\}, ~d= 4.
\end{equation} 
Since the output fidelity is degenerate with respect to the number of target channels $u$ and the Hamming distance, it is sufficient to only investigate one instance of each unique fidelity.
Hence, it is possible to explore all the unique output fidelities $F_{[u:u|d]}^{k_{\max}}$ by keeping the channel pattern $\bs{i}$ fixed, and monitoring how the counting functions of each sub-fidelity behave when the other pattern $\bs{i}^{\prime}$ is altered in such a way that changes the Hamming distance. This leads to recursive formulae for each counting function with respect to the Hamming distance which can each be easily solved.

In the following we consider two initially identical $m$-length channel patterns $\bs{i},\bs{i}^{\prime}$ which are $u$-CPF, and derive recurrence relations for the counting functions by keeping $\bs{i}$ fixed and editing the pattern $\bs{i}^{\prime}$.

\begin{itemize}
\item The counting function $f_{1:1}(u,u,d)$ monitors the number of times the sub-fidelity $F_{1:1}$ occurs within the total output fidelity expression. When $d$ is increased by 2, this means that a single background ($0$) and target ($1$) channel in the pattern $\bs{i}^{\prime}$ have been exchanged. Since every channel is probed in conjunction with every other channel, then there will be an increase of one pair of sub-patterns $\bs{\nu}_{\bs{i}_{\bs{s}}}= \{0,1\}$ and $\bs{\nu}_{\bs{i}_{\bs{s}}^{\prime}}  = \{1,0\}$ between the modified images, due to an entangled probe pair irradiated over the exchanged indices. Since this exchange happens nowhere else in the pattern, it is the only instance. Indeed, any additional exchange of this kind increments the number of times that this sub-pattern pair occurs. This corresponds to the recursion relation,
\begin{equation}
f_{1:1}(u,u,d) = f_{1:1}(u,u,d-2) + (d-1).
\end{equation}
\item A similar behaviour is shown for the counting function $f_{0:2}(u,u,d)$. The sub-fidelity $F_{0:2}$ only occurs when we find pairs of sub-patterns $\bs{\nu}_{\bs{i}_{\bs{s}}} = \{0,0\}$ and $\bs{\nu}_{\bs{i}_{\bs{s}}^{\prime}} = \{1,1\}$ between the modified patterns. After only one target channel exchange, there are still no $\{0,0\},\{1,1\}$ pairs, since this would require a Hamming distance of $d = 4$. But after two exchanges, we can find two of these pairs by irradiating entangled probes over both of the exchange sites. This behaviour corresponds to the recursion relation,
\begin{equation}
f_{0:2}(u,u,d) = f_{0:2}(u,u,d-2) + (d-2).
\end{equation}
\item The function  $f_{0:1}(u,u,d)$ describes the distribution of $\bs{\nu}_{\bs{i}_{\bs{s}}}= \{0,0\}$ and $\bs{\nu}_{\bs{i}_{\bs{s}}^{\prime}}  = \{0,1\}$ pairs of channel sub-patterns in the modified image space. Any $u$-CPF pattern has exactly $(m-u)$ background channels. When $d$ is initially increased by 2, there is a single background/target channel exchange; we label the channels involved in this as the \textit{exchange channels}. This leaves $(m-u-(d-1))$ background channels uninvolved in the exchange. It is then always possible to construct an entangled probe pair that is irradiated over an uninvolved background channel, and an exchange channel. Since there are two exchange channels, we conclude that there will be an increase of $2(m-u-(d-1))$ occurences of $\{0,0\},\{0,1\}$ sub-pattern pairs. This corresponds to the recursion relation, 
\begin{equation}
f_{0:1}(u,u,d) = f_{0:1}(u,u,d-2) + 2(m-u - (d - 1)).
\end{equation}
\item The function  $f_{1:2}(u,u,d)$ describes the distribution of $\bs{\nu}_{\bs{i}_{\bs{s}}}= \{0,1\}$ and $\bs{\nu}_{\bs{i}_{\bs{s}}^{\prime}}  = \{1,1\}$ sub-pattern pairs in the modified image space, and its recursive formula derived similarly to $f_{0:1}(u,u,d)$. A single, initial exchange involves two exchange channels. This leaves $(u - (d-1))$ target channels remaining throughout the pattern that have not been involved in the exchange. It is then always possible to irradiate an entangled probe pair over an uninvolved target channel, and an exchange channel. This will result in $2(u-(d-1))$ instances of $\{0,1\},\{1,1\}$ sub-pattern pairs, corresponding to the recursion relation,
\begin{equation}
f_{1:2}(u,u,d) = f_{1:2}(u,u,d-2) + 2(u -(d-1)).
\end{equation}
\end{itemize}
In all cases, $d = 2n, n\in \mathbb{N}$, and solving these recursion relations provides the functions,
\begin{align}
f_{1:1}(u,u,d) &= \frac{d^2}{4}\label{eq:f11_d},\\
f_{0:2}(u,u,d) &= \frac{d(d-2)}{4},\\
f_{0:1}(u,u,d) &= \frac{2d(m-u)-d^2}{2},\\
f_{1:2}(u,u,d) &= \frac{d(2u-d)}{2}. \label{eq:f12_d}
\end{align}
The full form fidelity can then be composed as the product
\begin{equation}
F_{[u:u|d]}^{k_{\max}} = F_{0:1}^{\frac{d[2(m-u) -d]}{2}}F_{1:1}^{\frac{d^2}{4}} F_{0:2}^{\frac{d(d-2)}{4}}F_{1:2}^{\frac{d(2u-d)}{2}},
\end{equation}
which provides a total characterisation of any fidelity that can occur when utilising a $k_{\max}$-LNN dynamic block protocol over strictly $u$-CPF image spaces. 

\subsection{$u/v$-CPF Output Fidelities}
To generalise the previous result, we need to consider the case of potentially unequal CPF image spaces (non-identical numbers of target channels in each channel pattern) and thus provide a characterisation of the fidelity of any pattern pair from different image spaces. Given $\bs{i} \in {\mc{U}_{{\text{CPF}}}^u}$ and $\bs{i}^\prime \in {\mc{U}_{{\text{CPF}}}^v}$ and assuming that $u < v$, all possible Hamming distances take the values:
\begin{gather}
\begin{gathered}
d = d_h(\bs{i}, \bs{i}^\prime) = 2t - (u+v), \\ \forall 
t \in \{v,\ldots,\min\{u+v,m\}\}.
\end{gathered} 
\label{eq:Ham_uv}
\end{gather}
This generates bounds on the Hamming distance,
\begin{gather}
\begin{gathered}
d_{\text{min}} \leq d \leq d_{\text{max}},\\
d_{\text{min}} \defeq v-u,~~
d_{\text{max}} \defeq \min\left\{ u+v, 2m-(u+v)\right\}, 
\end{gathered}
\label{eq:Ham_bounds}
\end{gather}
Therefore the Hamming distance may now have odd or even parity based on the values of $u/v$ unlike in the previous case of strict $u$-CPF. Furthermore, it is clear that if we consider two patterns $\bs{i} \in {\mc{U}_{{\text{CPF}}}^u}$ and $\bs{i}^\prime \in {\mc{U}_{{\text{CPF}}}^v}$, then the minimum Hamming distance will never be zero, but some non-zero value. Hence, it can be easily shown that the recursive functions take the following minimum Hamming distance values,
\begin{align}
\begin{aligned}
f_{0:1}(u,v,d_{\text{min}}) &= (m-v)(v-u),\\
f_{1:1}(u,v,d_{\text{min}}) &= 0, \\
f_{0:2}(u,v,d_{\text{min}}) &= \frac{1}{2}(v-u-1)(v-u),\\
f_{1:2}(u,v,d_{\text{min}}) &= u(v-u),
\end{aligned}
\end{align}
which clearly collapse to zero when $v=u$, as before. The recursive relationships then follow intuitively from the arguments of the previous section,
\begin{align}
\begin{aligned}
f_{0:1}(u,v,d) &= f_{0:1}(u,v,d-2) + 2m-(u+v) -2(d-1),\\
f_{1:1}(u,v,d) &= f_{1:1}(u,v,d-2) + (d-1), \\
f_{0:2}(u,v,d) &= f_{0:2}(u,v,d-2) + (d-2), \\
f_{1:2}(u,v,d) &= f_{1:2}(u,v,d-2) + (u+v) -2(d-1).
\end{aligned}
\end{align}
These recursive functions are simple to solve with the initial conditions, and we find the following generalised sub-fidelity exponents for images drawn from $u/v$-CPF image spaces:
\begin{align}
&f_{0:1}(u,v,d) =   \frac{2d(m-v) - d^2 + d_{\text{min}}}{2} ,  \label{eq:f01_uv} \\
&f_{1:1}(u,v,d) = \frac{d^2-d_{\text{min}}^2}{4},\\
&f_{0:2}(u,v,d) =  \frac{d(d-2)+d_{\text{min}}^2}{4} ,\\
&f_{1:2}(u,v,d) =  \frac{d(2v - (d-d_{\text{min}})) - 2d_{\text{min}}^2}{2}  \label{eq:f12_uv}.
\end{align}
These can be seen to be generalisations of the strictly $u$-CPF exponents where we account for different distance ranges using $d_{\text{min}}$, and can be re-expressed in terms of Eqs.~(\ref{eq:f11_d})-(\ref{eq:f12_d}). The output fidelity with respect to Hamming distance can then be expressed for any $u$ and $v$,
\begin{align}
\begin{aligned}
F_{[u:v|d]}^{k_{\max}} = &F_{0:1}^{\frac{2d(m-v) - d^2 + d_{\text{min}}}{2}}F_{1:1}^{ \frac{d^2-d_{\text{min}}^2}{4}} \\
\times &F_{0:2}^{\frac{d(d-2)+d_{\text{min}}^2}{4}}F_{1:2}^{\frac{d(2v - (d-d_{\text{min}})) - 2d_{\text{min}}^2}{2}}.
\end{aligned}
\end{align}

\subsection{Generalised CPF Error Bounds}

\subsubsection*{Error Bounds for $u$-CPF}

Consider the image space of $m$-length channel patterns that possess precisely $u$ target channels, i.e.~the $u$-CPF image space ${\mc{U}_{{\text{CPF}}}^u}$. For pairs of channel patterns from this image space, all the possible Hamming distances take the values
\begin{equation}
d_{h}(\bs{i}, \bs{i}^\prime) = 2(t - u), \> \forall t \in \{u+1,\ldots, 2u\}.
\end{equation}
Using the fidelity degeneracy properties of $k_{\max}$-LNN output states, the total error quantity can be written as a sum over all unique Hamming distances within the $u$-CPF image space,
\begin{align}
D_{u:u}^{\text{\tiny $k_{\text{\tiny max}}$}} [M] &= \sum_{\bs{i}\neq\bs{i}^{\prime} \in \mc{U}_{\text{CPF}}^u}  F^M(\bs{\varphi}_{\bs{\nu_{i}},\mu}, \bs{\varphi}_{\bs{\nu}_{\bs{i}^{\prime}},\mu}), \\
&=  \sum_{t=u+1}^{2u} {C_m^t} C_t^u C_u^{2u-t} ({F}_{[u:u|2(t-u)]}^{k_{\max}} )^{M} \label{eq:uCPF}.
\end{align}
This result follows from the use of the unique fidelity function defined in Eq.~(\ref{eq:kmaxUFid}), and basic counting arguments over the space of all $u$-CPF channel patterns for which we refer the reader to Appendix E, Ref.~\cite{ThermalPatt} for more details. 
Using this formula in Eq.~(\ref{eq:uCPF}) then allows us to write the following error bounds,
\begin{equation}
\frac{1}{2(C_{m}^u)^2}D_{{u:u}}^{\text{\tiny $k_{\text{\tiny max}}$}} [2M]  \leq p_{\text{err}} \leq \frac{1}{C_{m}^u} D_{{u:u}}^{\text{\tiny $k_{\text{\tiny max}}$}} [M].
\end{equation}
These bounds are easily confirmed by setting $u=1$ and $d=2$, which reproduces  Eqs.~(\ref{eq:CPF_ifLB})-(\ref{eq:CPF_ifUB}) immediately.\par

\subsubsection*{Error Bounds for Bounded CPF}

We may also consider more general image spaces based on the CPF formalism. Introduced in \cite{ThermalPatt}, one can study Bounded CPF (BCPF), which refers to a scenario in which there is ambiguity over the number of targets present in each channel pattern. That is, we can consider a larger image space which is the union of a number of CPF image spaces, 
\begin{equation}\mc{U}_{\text{CPF}}^{\bs{u}} = \bigcup_{u\in\bs{u}}  \mc{U}_{\text{CPF}}^u
\end{equation}
where $\bs{u}$ contains all possible numbers of target channels in any image in the space. If $\bs{u} = \{0,\ldots,m\}$ then the total image space $\mc{U}_{\text{CPF}}^{\bs{u}}$ is the complete set of all binary channel patterns. Similarly one may define $\bs{u} = \{1\}$ such that there is no ambiguity over the number of target channels, and we regain the single target channel CPF image space.\par

In this setting, the total error quantity can be decomposed into two contributions: Errors that are generated from misclassifying channel patterns from the same image space, and errors that are generated from misclassifying channel patterns from different image spaces. This can be more precisely written as, 
\begin{align}
D_{{\bs{u}}}^{\text{\tiny $k_{\text{\tiny max}}$}} [M]  &= \sum_{\bs{i}\neq\bs{i}^{\prime} \in \mc{U}_{\text{CPF}}^{\bs{u}}}  F^M(\bs{\varphi}_{\bs{\nu_{i}},\mu}, \bs{\varphi}_{\bs{\nu}_{\bs{i}^{\prime}},\mu}), \\
&= \sum_{j \in \bs{u} }D_{u:u}^{\text{\tiny $k_{\text{\tiny max}}$}} [M]  + \sum_{ u \in \bs{u} }   \sum_{v \neq u} \tilde{D}_{u:v}^{\text{\tiny $k_{\text{\tiny max}}$}}[M].
\label{eq:Boundk3}
\end{align}
Here $\tilde{D}_{\bs{F}_2}^{\text{\tiny $k_{\text{\tiny max}}$}}[M]$ is sum over all unequal fidelities between $u$-CPF and $v$-CPF patterns when $u\neq v$, which describes a contribution to the error probability from the misclassification of channel patterns from different CPF image spaces. Let $u \leq v$ count the number of targets in each image space. Then the valid Hamming distance between any two patterns is given by the Eqs.~(\ref{eq:Ham_uv}) and (\ref{eq:Ham_bounds}). Once again, following the counting arguments in \cite{ThermalPatt}, it can be shown that this error contribution can be re-parameterised in terms of the Hamming distance and the degenerate output fidelities,
\begin{equation}
 \tilde{D}_{u:v}^{\text{\tiny $k_{\text{\tiny max}}$}} [M] = \hspace{-3mm}\sum_{t=v}^{\min\{u+v,m\}} \hspace{-3mm}C_m^t C_t^v C_v^{(u+v)-t} ({F}_{[u:v|2t-(u+v)]}^{k_{\max}} )^{M} 
\end{equation}
With these results at hand, and letting $\Sigma \defeq \sum_{u\in\bs{u}} C_{m}^u$ be the uniform \textit{a priori} probability of each channel pattern, the error probability associated with $\bs{u}$-BCPF using the $k_{\max}$-LNN unassisted discrimination protocol is given by
\begin{equation}
\frac{1}{2\Sigma^2} D_{\bs{u}}^{\text{\tiny $k_{\text{\tiny max}}$}} [2M]  \leq p_{\text{err}} \leq \frac{1}{\Sigma} D_{\bs{u}}^{\text{\tiny $k_{\text{\tiny max}}$}} [M] . \label{eq:kBCPF_Bounds}
\end{equation}

\section{Classical Output Fidelities\label{sec:ClassF}}
The best known classical multi-channel discrimination protocols are achieved via single-mode coherent states. For quantum reading, these take the form 
$
\bigotimes_{i=1}^m \ket{{N_S}}\!\bra{N_S}_i
$. Probing a single pure loss channel with either $\eta_B$ or $\eta_T$ transmissivity, the fidelity between the two possible output states is given by \cite{PatternRecog},
\begin{equation}
F_{\text{cl}}^{\text{\scriptsize loss}} = \exp\left[-\frac{N_S}{2}(\eta_B - \eta_T)^2\right].
\end{equation}\par
For additive noise channels, the optimal classical input state is just the $m$-copy vacuum state $\ket{0}^{\otimes m}$, since displacements or phase shifts have no impact on the output states from the channel. Probing a single additive-noise channel with either  $\nu_B$ or $\nu_T$ noise, the fidelity between the two possible output states is given by \cite{OptEnvLoc},
\begin{equation}
F_{\text{cl}}^{\text{\scriptsize add}} = (\sqrt{(\nu_T +1)(\nu_B +1)} - \sqrt{\nu_T \nu_B})^{-1}.
\end{equation}

\begin{widetext}
\section{Two-Mode Sub-Fidelities\label{sec:SubFids}}
The two mode output sub-fidelities over binary channel patterns relevant to quantum reading and environment localisation (for additive noise channels) can be derived through the use of Gaussian fidelity formula from \cite{GFid}. 
For quantum reading purposes, define $\eta_j$ for $j\in \{B,T\}$ background or target transmissivity, and $N_S$ as the mean photon number of the incident probe. Then the sub-fidelities take the form,
\begin{align}
F_{0:1}(N_S,\eta_T,\eta_B) &= F_{1:2}(N_S,\eta_B,\eta_T) = \frac{-\sqrt{
\theta - | \phi | +
|\psi|} -
\sqrt{\theta+ | \phi |+|\psi| }}{\sqrt{2}\phi},\\
F_{1:1}(N_S,\eta_T,\eta_B) &= \frac{1}{N_S \left(-2 \sqrt{\bar{\eta}}+\eta_{*}\right)+1},\\
F_{0:2}(N_S,\eta_T,\eta_B) &\defeq \frac{2 N_S \sqrt{\bar{\eta}} + \sqrt{N_S (\tilde{\eta} - 4 N_S\bar{\eta}) -1}}{{1-N_S\tilde{\eta}}},
\end{align}
where the following quantities are defined,
\begin{align}
\bar{\eta} &= \eta_B\eta_T (1-\eta_T) (1-\eta_B), \>\> \tilde{\eta} = (\eta_B+\eta_T-2) (\eta_B+\eta_T),\>\> \eta_{*} = \eta_B+\eta_T -2 \eta_B \eta_T,\\
\theta &= 4 \sqrt{\eta_B^3 \eta_T} N_S^2 \sqrt{(\eta_B-1)^3 (\eta_T-1)}, \\
\tilde{\theta}&= 2 \sqrt{\eta_B^3 \eta_T} N_S(N_S+1)  -  \eta_TN_S - 1, \\
\phi &=  \tilde{\theta} - \eta_B^2 N_S (N_S-1) - \eta_B N_S (\eta_T (N_S-1)+3),\\
\psi &= \tilde{\theta} + \eta_B\eta_*N_S^2(2\eta_B -3) + \eta_BN_S\left( \tilde{\eta} +2\eta_B\eta_T - 3\right).
\end{align}
For the purposes of environment localisation of additive noise channels, denoting the additive noise properties $\nu_j$ for $j\in \{B,T\}$, the sub-fidelities take the form,
\begin{align}
F_{0:1}(N_S,\nu_T,\nu_B) &= F_{1:2}(N_S,\nu_B,\nu_T) = \frac{\sqrt{{
\nu_B \omega_B \left(\bar{\gamma}+\sqrt{\gamma_{B,T}\gamma_{T,B}}\right)}
}+{\sqrt{1 + \bar{\gamma} +  \nu_B \omega_B (\bar{\gamma} +2+ \sqrt{\gamma_{B,T}\gamma_{T,B}})}}}{\bar{\gamma}+ 2\nu_B\omega_B + 1}\\
F_{1:1}(N_S,\nu_T,\nu_B) &= \frac{1}{2\left[ \tilde{\nu}_{T,B} -\nu_T - \sqrt{\tilde{\nu}_{T,B}\> \tilde{\nu}_{B,T} }\right] +1}\\
F_{0:2}(N_S,\nu_T,\nu_B) &= \frac{1}{\sqrt{2 \nu_B \omega_{T} (2\nu_T \omega_B +1)+\nu_T (\omega_T +2\mu ) +\nu_B+1} -2 \sqrt{\nu_B \nu_T\omega_B \omega_T}}
\end{align}
where we define the quantities,
\begin{align}
\mu &= N_S+\frac{1}{2},\\
  \omega_{j}& = \nu_j + 2\mu,\\
   \tilde{\nu}_{i,j} &= N_S(\nu_i+\nu_j)+ \nu_i(\nu_j  +1), \\
\gamma_{i,j} &= 2 \mu  (\nu_i+\nu_j)+ \nu_i(2 \nu_j+1)-\nu_j, \>\> \bar{\gamma} = \frac{\gamma_{B,T}  +\gamma_{T,B}}{2}.
\end{align}

\end{widetext}
\hspace{1mm}

\end{document}